\documentclass[10pt,reqno]{amsart}
\usepackage[utf8]{inputenc}
\usepackage{a4wide}
\usepackage[T1]{fontenc}
\usepackage{amssymb}
\usepackage[english]{babel}
\usepackage{a4wide}
\usepackage{xcolor}
\usepackage{enumitem}
\usepackage{bbm}
\usepackage[pagebackref, colorlinks = true, linkcolor = teal, urlcolor  = teal, citecolor = violet]{hyperref}
\usepackage{upgreek}
\usepackage{graphicx}

\newtheorem{thm}{Theorem}[section]
\newtheorem{prop}[thm]{Proposition}

\newtheorem{cor}[thm]{Corollary}
\theoremstyle{plain}
\newtheorem{lem}[thm]{Lemma}
\theoremstyle{definition}

\theoremstyle{remark}

\newtheorem{rmq}{Remark}
\numberwithin{equation}{section}

\newcommand{\tr}{\operatorname{tr}}

\newcommand{\C}{\mathbb{C}}
\newcommand{\R}{\mathbb{R}}

\newcommand{\pp}{\mathbb{P}}
\newcommand{\pprho}[1]{\pp^{#1}}
\newcommand{\qq}{\mathbb{Q}}
\newcommand{\qumu}[1]{\qq_{#1}}
\newcommand{\ee}{\mathbb{E}}
\newcommand{\eerho}[1]{\ee^{#1}}
\newcommand{\eech}{\ee\ch}
\newcommand{\eemu}[1]{\ee_{#1}}

\newcommand{\cB}{\mathcal{B}}

\newcommand{\cD}{\mathcal{D}}
\newcommand{\cF}{\mathcal{F}}
\newcommand{\sys}{\mathcal{S}}
\newcommand{\env}{\mathcal{E}}

\newcommand{\ii}{\mathrm{i}}
\newcommand{\dd}{\mathrm{d}}
\newcommand{\e}{\mathrm{e}}
\newcommand{\Lin}{\mathcal{L}}

\newcommand{\ind}{\mathbbm{1}}

\newcommand{\Id}{\mathrm{Id}}
\renewcommand{\and}{\mbox{ and }}
\newcommand{\inv}{_{\mathrm{inv}}}
\newcommand{\ch}{^{\mathrm{ch}}}
\newcommand{\rhoinv}{\rho\inv}
\newcommand{\pc}[1]{\mathrm{P}\C^{#1}}
\newcommand{\mc}[1]{\mathrm{M}_{#1}(\C)}
\newcommand{\lerg}{\hyperlink{lerg}{\bf ($\Lin$-erg)}}
\newcommand{\pur}{\hyperlink{pur}{\bf(Pur)}}
\newcommand{\nscpur}{\hyperlink{nscpur}{\bf(NSC-Pur)}}
\newcommand{\Lp}[1]{\mathrm{L}^{#1}}
\newcommand{\filtreP}{\big(\Omega,(\cF_t)_t,\pp\big)}
\renewcommand{\tfrac}{\genfrac{}{}{}1}
\newcommand{\ran}{\mathop{\rm ran}\nolimits}
\newcommand{\x}{\mathcal X}
\newcommand{\y}{\mathcal Y}
\newcommand{\z}{\mathcal Z}

\usepackage{xargs}

\title{Invariant Measure for Stochastic Schr\"odinger Equations}

\author{T.\ Benoist}
\address{Institut de Math\'ematiques de Toulouse, UMR5219, Universit\'e de Toulouse, CNRS, UPS IMT, F-31062 Toulouse Cedex 9, France}
\email{tristan.benoist@math.univ-toulouse.fr}
\author{M.\ Fraas}
\address{Department of Mathematics, Virginia Tech, Blacksburg, VA 24061, United States of America}
\email{fraas@vt.edu}
\author{Y.\ Pautrat}
\address{Laboratoire de Math\'ematiques d'Orsay, Univ.\ Paris-Sud, CNRS, Universit\'e Paris-Saclay, 91405 Orsay, France}
\email{yan.pautrat@math.u-psud.fr}
\author{C.\ Pellegrini}
\address{Institut de Math\'ematiques de Toulouse, UMR5219, Universit\'e de Toulouse, CNRS, UPS IMT, F-31062 Toulouse Cedex 9, France}
\email{clement.pellegrini@math.univ-toulouse.fr}

\date{\today}

\begin{document}
	
\begin{abstract}
Quantum trajectories are Markov processes that describe the time-evolution of a quantum system undergoing continuous indirect measurement. Mathematically, they are defined as solutions of the so-called ``Stochastic Schr\"odinger Equations'', which are nonlinear stochastic differential equations driven by Poisson and Wiener processes. This paper is devoted to the study of the invariant measures of quantum trajectories. Particularly, we prove that the invariant measure is unique under an ergodicity condition on the mean time evolution, and a  ``purification'' condition on the generator of the evolution. We further show that quantum trajectories converge in law exponentially fast towards this invariant measure. We illustrate our results with examples where we can derive explicit expressions for the invariant measure. 
\end{abstract}

\maketitle

\setcounter{tocdepth}{1}
\tableofcontents
\section{Introduction}
Under a Markov approximation, the evolution of an open quantum system $\sys$ in interaction with an environment $\env$ is described by the Gorini--Kossakowski--Sudarshan--Lindblad Master (GKSL) equation \cite{GKS,Li}. More precisely, assuming that the system is described by the Hilbert space $\C^k$, the set of its states is defined as the set $\cD_k$ of density matrices, i.e.\ positive semidefinite matrices with trace one:
\[\cD_k=\{\rho\in \mc k \mbox{ s.t. } \rho\geq0,\tr \rho=1\}.\]
The evolution $t\in\R_+\mapsto \bar\rho_t\in\cD_k$ of states of the system is then determined by the GKSL equation (also called quantum master equation):
\begin{equation}\label{eq:lindblad}
\dd\bar\rho_t=\Lin(\bar\rho_t)\,\dd t,\quad\bar\rho_0\in\cD_k,
\end{equation}
where $\Lin$ is a bounded linear operator on $\mc k$ of the form
\begin{equation}\label{eq_deflindblad}
	\Lin: \rho\mapsto-\ii[H,\rho]+\sum_{i\in I}\big( V_i\rho V_i^*-\tfrac{1}{2}\{V_i^*V_i,\rho\}\big),
\end{equation}
with $I$ a finite set, $H\in \mc k$ self adjoint, and $V_i\in \mc k$ for each $i\in I$ ($[\cdot,\cdot]$ and $\{\cdot,\cdot\}$ are respectively the commutator and anticommutator). Such an $\Lin$ is called a Lindblad operator.

Since $\Lin$ is linear, $t\mapsto \bar\rho_t$ is given by $\bar\rho_t=\e^{t\Lin}(\bar\rho_0)$. The flow is therefore a semigroup $(\e^{t\Lin})_t$, which consists of completely positive, trace-preserving maps (see \cite{wolftour}). In particular, $\Lin$ is the generator of a semigroup of contractions, thus $\operatorname{spec}\mathcal{L}\subset \{\lambda\in\C\ \mbox{ s.t.}\ \operatorname{Re}\lambda\leq 0\}$. Since $\e^{t\Lin}$ is trace preserving, $0\in \operatorname{spec}\Lin$. The following assumption is equivalent to the simplicity of the eigenvalue~$0$  \cite[Proposition 7.6]{wolftour}:

\hypertarget{lerg}{}
\begin{description}\item[\lerg] There exists a unique non zero minimal orthogonal projection $\pi$ such that $\Lin(\pi \mc k\pi)\subset \pi \mc k\pi$.
\end{description}

Assumption {\lerg} implies directly that there exists a unique $\rhoinv\in\cD_k$ such that $\Lin\rhoinv=0$. Moreover, one can show that {\lerg} implies the existence of $\lambda>0$ such that for any $\rho\in\cD_k$, $\e^{t\Lin}(\rho)=\rhoinv+O(\e^{-\lambda t})$ (see \cite[Proposition 7.5]{wolftour}).

The above framework generalizes that of continuous-time Markov semigroups on a finite number of sites: density matrices $\rho$ over $\C^k$ generalize probability distributions over $k$ classical states, while Lindbladians $\Lin$ generalize generators of Markov jump processes. In Section \ref{sec:classical_MC}, we show how a classical finite state Markov jump process can be encoded in the present formalism.
\smallskip

The family $(\bar \rho_t)_t$ describes the reduced evolution of the system $\sys$ when coupled to an environment $\env$ in a conservative manner. This evolution can be derived by considering the full Hamiltonian of $\sys+\env$ in relevant limiting regimes, e.g.\ the \emph{weak coupling} or \emph{fast repeated interactions} regimes, and tracing out the environment degrees of freedom (see \cite{davies_markovian_1974,davies_markovian_1976} and \cite{attal2006repeated} respectively). It can also be described by a \emph{stochastic unravelling}, i.e.\ a stochastic process $(\rho_t)_t$ with values in $\cD_k$ such that the expectation $\overline \rho_t$ of $\rho_t$ satisfies \eqref{eq:lindblad}; this method was developed in \cite{Ba2,BaGre,Ba3}. One possible choice of a stochastic unravelling is described by the following stochastic differential equation (SDE), called a \emph{stochastic master equation}:
\begin{equation}\label{eq1:SME}
\begin{aligned}
\dd\rho_t=&\Lin(\rho_{t-})\,\dd t\\
&+\quad\sum_{i\in I_b}\Big(L_i\rho_{t-}+\rho_{t-} L_i^*-\tr\big(\rho_{t-}(L_i+L_i^*)\big)\rho_{t-}\Big)\,\dd B_i(t)\\
&+\quad\sum_{j\in I_p} \Big(\frac{C_j\rho_{t-}C_j^*}{\tr(C_j\rho_{t-}C_j^*)}-\rho_{t-}\Big)\Big(\dd N_j(t)-\tr(C_j\rho_{t-} C_j^*)\,\dd t\Big), 
\end{aligned}
\end{equation}
where
\begin{itemize}
	\item $I=I_b\cup I_p$ is a partition of $I$ such that $L_i=V_i$ for $i\in I_b$ and $C_j=V_j$ for $j\in I_p$,
	\item each $B_i$ is a Brownian motion,
	\item each $N_j$ is a Poisson process of intensity $t\mapsto\int_0^t \tr(C_j\rho_{s-}C_j^*)\dd s$.
\end{itemize}
\begin{rmq}\label{eq_esprhot}
The processes $\big(B_j(t)\big)_t$ and $\left(N_j(t)-\int_0^t \tr(C_j\rho_{s-}C_j^*)\dd s\right)_t$ are actually martingales. Then assuming that \eqref{eq1:SME} accepts a solution, it is easy to check that for any $t\geq 0$, the expectation of $\rho_t$ is equal to $\bar\rho_t$ whenever $\rho_0=\bar\rho_0$.
\end{rmq}

Proper definitions of these Poisson processes and proofs of existence and uniqueness of {the solution to \eqref{eq1:SME} can be found in \cite{BaGre,Ba3,Pe1,Pe2,Pe3}. A solution $(\rho_t)_t$ of Equation \eqref{eq1:SME} is called a \emph{quantum trajectory}.

Equations of the form \eqref{eq1:SME} are used to model experiments in quantum optics (photo-detection, heterodyne or homodyne interferometry), particularly for measurement and control (see \cite{Ca1,Ha1,wisemanmilburn}). They were also introduced as stochastic collapse models (see \cite{diosi_quantum_1988,gisin_quantum_1984}) and as numerical tools to compute $\overline\rho_t$ (see \cite{dalibard_wave-function_1992}). Here we are interested in the fact that they model the evolution of the system $\sys$ when continuous measurements are done on the environment~$\env$. This can be shown starting from quantum stochastic differential equations using quantum filtering \cite{Ba5,Belavkin1992,bouten_quantum_2006,Ga1,Partha1}. An approach using the notion of a priori and a posteriori states has been also developed using ``classical'' stochastic calculus (see the reference book by Barchielli and Gregoratti \cite{BaGre}, and references therein). Continuous-time limits of discrete-time models can also be considered, see \cite{Pe1,Pe2,Pe3}. 

Equation \eqref{eq1:SME} has the property that if $\rho_0$ is an extreme point of $\cD_k$, then ${\rho}_t$ is almost surely an extreme point of $\cD_k$ for any $t\in\R_+$. Since we will extensively use this property, let us make it explicit. The extreme points of $\cD_k$ are the rank-one orthogonal projectors of $\C^k$; for any $x\in\C^k\setminus\{0\}$, let $\hat x$ be its equivalence class in $\pc k$, the projective space of $\C^k$. For $\hat x\in \pc k$, let $\pi_{\hat x}$ be the orthogonal projector onto $\C x$. Then $\hat x\in \pc k\mapsto\pi_{\hat x}$ is a bijective map from $\pc k$ to the set of extreme points. Assume now that $\rho_0=\pi_{\hat x_0}$ for some $\hat x_0\in \pc k$. Then it is easy to check that $\rho_t=\pi_{\hat x_t}$ almost surely for any $t\in\R_+$, with $t\mapsto x_t$ the unique solution to the following SDE, called a \emph{stochastic Schr\"odinger equation}:
\begin{equation}\label{eq:SDEpur}
\begin{aligned}
\dd x_t&=D(x_{t-})x_{t-}\,\dd t\\
&\quad+\sum_{i\in I_b} \big(L_i-\tfrac{1}{2}v_i(t-)\,\Id\big)x_{t-}\,\dd B_i(t)\\
&\quad+\sum_{j\in I_p}  \Big(\frac{C_j}{\sqrt{n_j(t-)}}-\Id\Big) x_{t-}\,\dd N_j(t),
\end{aligned}
\end{equation}
for $x_0\in\hat x_0$ of norm one, where the operator $D(x_{t-})$ is defined as
\[D(x_{t-})=-\big(\ii H+\frac{1}{2}\sum_{i\in I_b} L_i^*L_i+\frac{1}{2}\sum_{j\in I_p}  C_j^*C_j\big)+\frac{1}{2}\sum_{i\in I_b} v_i(t-)\big(L_i-\tfrac{1}{4}\,v_i(t-)\,\Id\big)+\frac{1}{2}\sum_{j\in I_p}  n_j(t-),\]
with
\begin{equation*}
v_i(t-)=\langle x_{t-},(L_i+L_i^*)x_{t-}\rangle,\qquad
n_j(t-)=\langle x_{t-},C_j^*C_j x_{t-} \rangle=\Vert C_j x_{t-}\Vert^2.
\end{equation*}
The brackets $\langle \cdot, \cdot \rangle$ denote the scalar product in $\mathbb{C}^k$. Without possible confusion, a solution $(x_t)_t$ will be also called a quantum trajectory. Remark that $\|x_0\|=1$ implies $\|x_t\|=1$ almost surely for any $t\in\R_+$; remark also that the numerical computation of $ \rho_t$ involves only multiplications of matrices with vectors and not multiplications of matrices (this is the motivation for the use of quantum trajectories as numerical tools mentioned above).

In the physics literature, extreme points of $\cD_k$ are called pure states. In particular, the preceding paragraph shows that the evolution dictated by Eq.\ \eqref{eq1:SME} preserves pure states. It actually has also the property that quantum trajectories (solution of \eqref{eq1:SME}) tend to ``purify''. This has been formalized by Maassen and K\"ummerer in \cite{Maassen} for discrete-time quantum trajectories, and extended to the continuous-time case by Barchielli and Paganoni in \cite{Ba4}. Purification is related to the following assumption (here $A\propto B$ means there exists $\lambda\in\C$ such that $A=\lambda B$ or $\lambda A=B$. Particularly we allow for $\lambda=0$).
\hypertarget{pur}{}
\begin{description}
	\item[\pur] Any non zero orthogonal projector $\pi$ such that for all $i\in I_b$, $\pi (L_i+L_i^*)\pi \propto  \pi$  and for all $j\in I_p$, $\pi C_j^*C_j\pi\propto \pi$ has rank one.
\end{description}
As shown in \cite{Ba4}, {\pur} implies that for any $\rho_0\in\cD_k$
\begin{equation}\label{eq:purification}
\lim_{t\to\infty} \inf_{\hat y\in \pc k}\|\rho_t -\pi_{\hat y}\|=0\quad \text{almost surely.}
\end{equation}

The main goal of this article is to show how the exponential convergence of the solution $(\overline \rho_t)_t$ of Eq.\ \eqref{eq:lindblad}, induced by {\lerg}, translates for its stochastic unravelling $(\rho_t)_t$ solution of Eq.\ \eqref{eq1:SME}. We prove uniqueness of the invariant measure for continuous-time quantum trajectories assuming both {\lerg} and \pur. From \eqref{eq:purification}, under these assumptions, the invariant measure will be concentrated on pure states, so we only need to prove uniqueness of the invariant measure for $(\hat x_t)_t$ equivalence class of $( x_t)_t$ solution  of \eqref{eq:SDEpur} (since $\pc k$ is compact and the involved process is Feller, the existence of an invariant measure is obvious). The difficulty of this proof lies in the failure of usual techniques like $\varphi$-irreducibility. Note that this question has already been partially addressed in the literature: essentially, only diffusive equations have been considered, i.e.\ equations for which Eq. \eqref{eq1:SME} or \eqref{eq:SDEpur} contain no jump term (in our notation, $I_p=\emptyset$). The results of \cite{Ba4} were, to our knowledge, the most advanced ones so far. In that article, algebraic conditions on the vector fields describing the stochastic differential equation are imposed to obtain the uniqueness of the invariant measure. This allows the authors to apply directly standard results from the analysis of stochastic differential equations. Unfortunately their assumptions are hard to check for a given family of matrices $(L_i)_{i\in I_b}$.

The main result of the present paper is the following theorem.
\begin{thm}\label{thm:expo_conv_wasser1}
Assume that {\pur} and {\lerg} hold. Then the Markov process $(\hat x_t)_t$ has a unique invariant probability measure $\mu\inv$, and there exist $C>0$ and $\lambda>0$ such that for any initial distribution $\mu$ of $\hat x_0$ over $\pc k$, for all $t\geq 0$, the distribution $\mu_t$ of $\hat x_t$ satisfies
\begin{equation*}
W_1(\mu_t,\mu\inv)\leq C\e^{-\lambda t}
\end{equation*}
where $W_1$ is the Wasserstein distance of order $1$.
\end{thm}

This theorem is more general than previous similar results in different ways. First, we consider stochastic Schr\"odinger equations involving both Poisson and Wiener processes. Second, our assumptions are standard for quantum trajectories and are easy to check for a given family of operators $\big(H,(L_i)_{i\in I_b},(C_j)_{j\in I_p}\big)$. Last, we prove an exponential convergence towards the invariant measure. As a byproduct, we also provide a simple proof of the purification expressed in Eq.\ \eqref{eq:purification} (see Proposition \ref{prop:pur}). To complete the picture, assuming only \pur, we show that {\lerg} is necessary. We also provide a complete characterization of the set of invariant measures of $(\hat x_t)$ whenever {\lerg} does not hold (see Proposition \ref{prop:set_inv_measure}). Arguments in Sections \ref{sec:inv} and \ref{app:inv_measure} are adaptations of \cite{Inv1}, where similar results for discrete-time quantum trajectories are considered. 
\medskip

The paper is structured as follows. In Section \ref{model}, we give a precise description of the model of quantum trajectories with a proper definition of the underlying probability space. In particular, we introduce a new martingale which is central to our proofs.  In Section \ref{sec:inv}, we prove Theorem \ref{thm:expo_conv_wasser1}. In Section~\ref{app:inv_measure} we derive the full set of invariant measures assuming only \pur. In Section \ref{sec:pur_not_nec} we show that {\pur} is not necessary even if {\lerg} holds. In Section \ref{sec:ex}, we provide some examples of explicit invariant measures. In Section \ref{sec:classical_MC} we provide an encoding of any classical finite state Markov jump process into a stochastic master equation.

\section{Construction of the model}\label{model}

\subsection{Construction of quantum trajectories}

In this section we fix the notations and introduce the probability space we use to study $(\hat x_t)_t$. First, for an element $x\neq 0$ of $\C^k$, and for an operator $A$ with $Ax\neq 0$ we denote
\[A\cdot \hat x=\widehat{Ax}.\]
We consider the following distance on $\pc k$:
\begin{equation} \label{eq_defdistance}
	d(\hat x,\hat y)=\sqrt{1- |\langle x,y\rangle|^2\,},
\end{equation}
for all $\hat x,\hat y\in \pc k$, where $x$ and $y$ are norm-one representatives of $\hat x$ and $\hat y$ respectively. 
We equip $\pc k$ with the associated Borel $\sigma$-algebra denoted by $\mathcal B$. 

Now we introduce a stochastic process with values in $\mc k$. Let $\filtreP$ be a filtered probability space with standard brownian motions $W_i$ for $i\in I_b$, and standard Poisson processes $N_j$ for $j\in I_p$, such that the full family $\big(W_i, N_j; \,i\in I_b, j\in I_p \big)$ is independent. The filtration $(\cF_t)_t$ is assumed to satisfy the standard conditions, and we denote $\cF_{\infty}$ by $\cF$ and the processes $\big(W_i(t)\big)_t$ and $\big(N_j(t)-t\big)_t$ are $(\cF_t)_t$-martingales under $\pp$. We denote by $\ee$ the expectation with respect to~$\pp$.

On $\filtreP$, for $s\in\R_+$, let $(S_t^s)_{t\in[s,\infty)}$ be the solution to the following SDE:
\begin{equation}\label{eq_defS}
\dd S_{t}^s= \big(K+{\textstyle {\frac{\# I_p}2}}\,{\Id}\big)S_{t-}^{s}\,\dd t+\sum_{i\in I_b} L_iS_{t-}^{s}\,\dd W_i(t)+\sum_{j\in I_p}(C_j-\Id)S_{t-}^{s}\,\dd N_j(t),\qquad S_s^s=\Id
\end{equation}
($\# I_p$ is the cardinal of $I_p$), where $$K=-\ii H-\frac{1}{2}(\sum_{i\in I_b} L_i^*L_i+\sum_{j\in I_p} C_j^*C_j).$$ Since standard Cauchy--Lipschitz conditions are fulfilled, the SDE defining $(S_t^s)_t$ has indeed a unique (strong) solution. We denote $S_t:=S_t^0$.
Note that for $s$ fixed the process $(S_t^s)_t$ is independent of $\cF_s$, and we have that for all $0\leq r\leq s\leq t$
\[S_t^sS_s^r=S_t^r.\]
In addition, for any $\rho\in\cD_k$, let $(Z_t^\rho)_t$ be the positive real-valued process defined by
\[Z_t^\rho=\tr(S_t^*S_t\rho),\]
and let $(\rho_t)_{t}$ be the $\cD_k$-valued process defined by
\[\rho_t=\frac{S_t\rho S_t^*}{\tr(S_t\rho S_t^*)}\]
if $Z_t^\rho\neq 0$, taking an arbitrarily fixed value whenever $Z_t^\rho=0$ (this value will always appear with probability zero in the sequel).

The following results on the properties of $(Z_t^\rho)_t$ were proven in \cite{Ba3}. We give short proofs adapted to our restricted setting where the Hilbert space is finite-dimensional, and $I=I_b\cup I_p$ is a finite set. 
\begin{lem}\label{lem_Zmg} For any $\rho\in\cD_k$, the stochastic process $(Z_t^\rho)_t$ is the unique solution of the SDE
\[\dd Z_t^\rho=Z_{t-}^\rho\Big(\sum_{i\in I_b}\tr\big((L_i+L_i^*)\rho_{t-}\big)\dd W_i(t)+\sum_{j\in I_p} \big(\tr(C_j^*C_j\rho_{t-})-1\big)\big(\dd N_j(t)-\dd t\big)\Big),\quad Z_0^\rho=1.\]
Moreover, $(Z_t^\rho)_{t}$ is a nonnegative martingale under $\pp$.
\end{lem}

\begin{proof}
The fact that $(Z_t^\rho)_t$ verifies the given SDE is a direct application of the It\^o formula. Since $(\rho_t)_t$ takes its values in the compact space $\cD_k$, that SDE verifies standard Cauchy--Lipschitz conditions, ensuring the uniqueness of the solution. Since the processes $\big(W_i(t)\big)_t$ and $\big(N_j(t)-t\big)_t$ are $\pp$-martingales, it follows that $(Z_t^\rho)_t$ is a $\pp$-local martingale. Since $\tr(C_j^*C_j\rho)\geq 0$ for any $j\in I_p$ and $\rho\in \cD_k$, and $(\rho_t)_t$ takes value in the compact space $\cD_k$, it follows from \cite[Theorem 12]{Kabanov} that $(Z_t^\rho)_{t\in [0,T]}$ is a $\pp$-nonnegative martingale for all $T$.
\end{proof}
For any $\rho\in\cD_k$, we define a probability $\pprho\rho_t$ on $(\Omega,\cF_t)$:
\begin{equation} \label{eq_defprho}
	\dd\pprho\rho_t=Z_t^\rho \,\dd\pp|_{\cF_t}.
\end{equation}
Since $(Z_t^\rho)_{t}$ is a $\pp$-martingale from Lemma \ref{lem_Zmg}, the family $(\pprho{\rho}_t)_t$ is consistent, that is $\pprho\rho_t(E)=\pprho\rho_s(E)$ for $t\geq s$ and $E\in\cF_s$. Kolmogorov's extension theorem defines a unique probability on $(\Omega,\cF_\infty)$, which we denote by $\pprho\rho$. We will denote by $\eerho\rho$ the expectation with respect to $\pprho\rho$.

The following proposition makes explicit the relationship between $\pp$ and $\pprho\rho$. It follows from a direct application of Girsanov's change of measure Theorem (see \cite[Theorems III.3.24 and III.5.19]{JacodShiryaev}). For all $i\in I_b$ and $t\in \R_+$, let
\[B_i^{\rho}(t)=W_i(t)-\int_0^t\tr\big((L_i+L_i^*)\rho_{s-}\big)\,\dd s.\]
\begin{prop}\label{prop:girsa} Let $\rho\in\cD_k$. Then, with respect to $\pprho\rho$, the processes $\{B_i^\rho\}_{i\in I_b}$ are independent Wiener processes and the processes $\{N_j\}_{j\in I_p}$ are point processes of respective stochastic intensity $\{t\mapsto \tr(C_j^*C_j\rho_{t-})\}_{j\in I_p}$.
\end{prop}

The process $(\rho_t)_t$ considered under $\pprho\rho$ models the evolution of a Markov open quantum system subject to indirect measurements. We refer the reader to \cite{BaGre,BreuerPetru,Ca1} and references therein for a more detailed discussion of this interpretation.
 
From It\^o calculus, $(\rho_t)_t$ is solution of the SDE
\begin{equation}
\begin{aligned}
\dd{\rho}_t&=\Lin(\rho_{t-})\dd t\\
&\quad+\sum_{i\in I_b}\Big(L_i\rho_{t-}+\rho_{t-} L_i^*-\tr\big(\rho_{t-}(L_i+L_i^*)\big)\rho_{t-}\Big)\dd B_i^\rho(t)\\\
&\quad+\sum_{j\in I_p} \Big(\frac{C_j\rho_{t-}C_j^*}{\tr(C_j\rho_{t-}C_j^*)}-\rho_{t-}\Big)\Big(\dd N_j(t)-\tr(C_j\rho_{t-} C_j^*) \,\dd t\Big).
\end{aligned}
\end{equation}
Proposition \ref{prop:girsa} then implies that with respect to $\pprho\rho$, the process $(\rho_t)_t$ is indeed the unique solution of \eqref{eq1:SME} with $\rho_0=\rho$. Similarly, if $\rho_0=\pi_{\hat {x}}$ for some $\hat x\in \pc k$, then with respect to $\pprho{\pi_{\hat x}}$, the process $\big(\frac{S_t {x}}{\|S_t{x}\|}\big)_t$ is the solution of \eqref{eq:SDEpur} with $x$ any norm one representative of $\hat {x}$.

Remark also that for any $\rho\in \mathcal D_k$, using \eqref{eq_defprho}, one has from Remark \ref{eq_esprhot}
\begin{equation} \label{exppsrhos}
	\ee(S_t\rho S_t^*) = \ee(\rho_t Z_t^\rho)=\eerho\rho(\rho_t)=\e^{t\Lin}(\rho).
\end{equation}

Our strategy of proof is based on the study of the joint distribution of $S_t$ and a random initial state $\hat x$. To this end, we consider the product  space $\Omega\times \pc k$ equipped with the filtration $(\cF_t\otimes\mathcal B)_t$ and the full $\sigma$-algebra $\cF\otimes\mathcal B$. For any probability measure $\mu$ on $\pc k$, and for all $E\in\cF$ and $A\in\mathcal B$, let
\begin{equation*}
\qumu\mu(E\times A)=\int\pprho{\pi_{\hat x}}(E)\,\ind_{\hat x \in A}\,\dd\mu(\hat x).
\end{equation*}
We will denote by $\eemu\mu$ the expectation with respect to $\qumu\mu$. Note that $\dd \pprho{\pi_{\hat x}}_t = \|S_t x\|^2 \,\dd \pp$ for any $\hat x\in \pc k$, so that $\pprho{\pi_{\hat x}}(\{S_tx=0\})=0$ for all $x\in \hat x$. Therefore
\[\qumu\mu\big(\{S_tx=0\}\big)=0\]
and there exists a process $(\hat x_t)_t$ for which
\[\hat x_t=S_t\cdot x\]
holds almost surely. It has the same distribution as the image by the map $x\mapsto \hat x$ of the solution $(x_t)_t$ to \eqref{eq:SDEpur} with $x_0\in \hat x$, $\|x_0\|=1$.

The following proposition shows that the laws of any $\cF$-measurable random variables are given by a marginal of $\qumu\mu$. For a probability measure $\mu$ on $\pc k$, we define $$\rho_\mu:=\mathbb E_\mu(\pi_{\hat x}).$$
\begin{prop}\label{prop:margin}
Let $\mu$ be a probability measure on $\pc k$, then $\rho_\mu \in \cD_k$ and for any $E\in\cF$,
\[\qumu\mu(E\times \pc k)=\pprho{\rho_\mu}(E).\]
\end{prop}
\begin{proof}
The fact that $\rho_\mu\in \cD_k$ follows from the positivity and linearity of the expectation. Concerning the second part, let $t\geq0$ and $E\in\cF_t$, then
\begin{equation*}
\qumu\mu(E\times \pc k)=\int\pprho{\pi_{\hat x}}(E)\, \dd\mu(x)=\int\!\int_{E}\tr( S_t^*S_t \pi_{\hat x}) \,\dd\pp \,\dd\mu(x).
\end{equation*}
Fubini's Theorem implies
\begin{equation*}
\qumu\mu(E\times \pc k)=\int_E\tr(S_t^*S_t\rho_\mu)\,\dd\mathbb P=\int_E Z_t^{\rho_\mu}\,\dd\pp=\pprho{\rho_\mu}_t(E).
\end{equation*}
The uniqueness of the extended measure in Kolmogorov's extension Theorem yields the proposition.
\end{proof}

\begin{rmq}
 Any $\cF$-measurable random variables $X$ can be extended canonically to a $\cF\otimes \mathcal B$-measurable random variables setting $X(\omega,\hat x)=X(\omega)$. Proposition \ref{prop:margin} then implies that the distribution of a $\cF$-measurable random variable under $\qumu\mu$ depends on $\mu$ only through $\rho_\mu$. The central idea of our proof is that assumption {\pur} will allow us to find a $\cF$-measurable process approximating $(\hat x_t)_t$. The $\cF$-measurability of the process will then imply that it inherits some ergodicity properties from assumption \lerg.
\end{rmq}

\begin{rmq} \label{rmq_muinvrhoinv}
If $\mu\inv$ is an invariant measure for the Markov chain $(\hat x_t)_t$, then with the above notation, $\rho_{\mu\inv}$ is an invariant state for $(\e^{t\Lin})_t$. In particular, if {\lerg} holds then $\rho_{\mu\inv}=\rho\inv$. This follows from the identities
\[\e^{t\Lin}(\rho_{\mu\inv}) = \int \e^{t\Lin}(\pi_{\hat x})\, \dd\mu\inv = \int S_t \pi_{\hat x} S_t^* \, \dd\pp \, \dd\mu\inv(\hat x)= \int \pi_{\hat x}\, \dd\mu\inv(\hat x)= \rho_{\mu\inv}\]
where the second identity uses \eqref{exppsrhos}.
\end{rmq}

\subsection{Key martingale}
The following process is the key to construct a $\cF$-measurable process approximating $(\hat x_t)_t$.
For any $t\geq 0$, let
\begin{equation}\label{eq:defmt}M_t=\frac{S_t^*S_t}{\tr(S_t^*S_t)},
\end{equation}
whenever $\tr(S_t^*S_t)\neq 0$, and give $M_t$ a fixed arbitrary value whenever $\tr(S_t^*S_t)=0$. Since, by definition, for any $\rho\in \cD_k$, $\pprho\rho\big(\{\tr S_t^*S_t=0\}\big)=0$, the arbitrary definition of $M_t$ on this set of vanishing probability is irrelevant. It turns out that with respect to $\pprho{\Id/k}$, $(M_t)_t$ is a martingale. For convenience we write $\pp\ch=\pprho{\Id/k}$ and similarly for any other $\rho$-dependent object, whenever $\rho=\Id/k$.

\begin{thm}\label{thm:M_t}
With respect to $\pp\ch$, the stochastic process $(M_t)_t$ is a bounded martingale. Therefore, it converges $\pp\ch$-almost surely and in $\Lp 1$ to a random variable $M_\infty$. Moreover, for any $\rho\in \cD_k$,
\[\dd\pprho\rho=k\,\tr(\rho M_{\infty})\, \dd\pp\ch,\]
and $(M_t)_t$ converges almost surely and in $\Lp 1$ to $M_{\infty}$ with respect to $\pprho\rho$.
\end{thm}

\begin{proof}
Expressing $(S_t)_t$ in terms of $B_i\ch$ for $i\in I_b$, we have that
\begin{equation*}
\dd S_{t}= \Big(K+{\textstyle {\frac{\# I_p}2}}\,{\Id}+\sum_{i\in I_b} \frac{\tr\big(S_{t-}^*(L_i+L_i^*)S_{t-}\big)}{\tr(S_{t-}^*S_{t-})}L_i\Big)S_{t-}\,\dd t+\sum_{i\in I_b} L_iS_{t-}\,\dd B_i\ch(t)+\sum_{j\in I_p}(C_j-\Id)S_{t-}\,\dd N_j(t).
\end{equation*}
Recall that the distributions of the $B_i\ch$ and $N_j$ under $\pp\ch$ are given by Proposition \ref{prop:girsa}.

Since $\tr(S_t^*S_t)$ is $\pp\ch$-almost surely non zero, we can define $R_t$ by $R_t=S_t/\sqrt{\tr(S_t^*S_t)}$ almost surely for $\pp\ch$, and therefore for $\pprho\rho$ and $\qumu\mu$. The It\^o formula implies
\begin{align*}
\dd M_t=&\sum_{i\in I_b} \Big(R_{t-}^*(L_i+L_i^*)R_{t-}-M_{t-}\tr\big(R_{t-}^*(L_i+L_i^*)R_{t-}\big)\Big)\,\dd B\ch_i(t)\\
&+\sum_{j\in I_p} \Big(\frac{R_{t-}^*C_j^*C_j R_{t-}}{\tr(R_{t-}^*C_j^*C_j R_{t-})}- M_{t-}\Big)\big(\dd N_j(t)-\tr(R_{t-}^* C_j^*C_jR_{t-})\,\dd t\big).
\end{align*}
Hence, with respect to $\pp\ch$,  $(M_t)_t$ is a local martingale. By definition, it is positive-semidefinite, and is also bounded since $\tr (M_t)=1$ almost surely. Thus $(M_t)_t$ is a martingale and standard theorems of convergence for martingales imply the convergence almost surely and in $\Lp 1$. 

By direct computation, we get $\dd\pprho\rho|_{\cF_t}=k\tr(\rho M_t)\, \dd\pp\ch|_{\cF_t}$. The $\Lp 1$ convergence of $(M_t)_t$ with respect to $\pp\ch$ then implies $\dd\pprho\rho=k\tr(\rho M_\infty)\,\dd\pp\ch$. Finally the inequality $\tr(AB)\leq \|A\|\tr(B)$ for any two positive semidefinite matrices implies $\pprho\rho\leq k\,\pp\ch$, which yields the $\Lp 1$ and almost sure convergence with respect to $\pprho\rho$.
\end{proof}

Now we are in the position to show that under the assumption {\pur} the limit $M_\infty$ is a rank-one projector. To this end let us introduce the polar decomposition of $(S_t)_t$: there exists a process $(U_t)_t$ with values in the set of $k\times k$ unitary matrices such that for all $t\geq0$
\[S_t=\sqrt{\tr(S_t^*S_t)}\, U_tM_t^{1/2}.\]

\begin{prop}\label{prop:pur} Assume that {\pur} holds. Then, for any $\rho\in \cD_k$, $\pprho\rho$-almost surely, the random variable $M_\infty$ is a rank-one orthogonal projector on $\C^k$.
\end{prop}

\begin{proof} First, since $\pprho\rho$ is absolutely continuous with respect to $\pp\ch$, proving the result with $\rho=\Id/k$ is sufficient. To achieve this, remark that the $\pp\ch$-almost sure convergence of $(M_t)_t$ and the $\pp\ch$-almost sure bound $\sup_{t\geq 0}\|M_t\|\leq 1$ imply the convergence of $\eech(M_t^2)$. Now recall that $R_t=S_t/\sqrt{\tr(S_t^*S_t)}$. The Itô isometry implies
\begin{align*}
\eech(M_t^2)=M_0^2&+\sum_{i\in I_b}\int_0^t\eech\Big( R_{s}^*(L_i+L_i^*)R_{s}-M_{s}\tr\big(R_{s}^*(L_i+L_i^*)R_s\big)\Big)^2\,\dd s\\
&+\sum_{j\in I_p}\int_0^t\eech\bigg(\Big(\frac{R_{s}^*C_j^*C_jR_{s}}{\tr(R_s^*C_j^*C_jR_s)}-M_{s}\Big)^2\tr(R_{s}^*C_j^*C_jR_s)\bigg)\,\dd s.
\end{align*}
Therefore, the convergence of $\eech(M_t^2)$ to $\eech(M_\infty^2)$ implies that
\[\int_0^\infty\eech\Big( R_{s}^*(L_i+L_i^*)R_{s}-M_{s}\tr\big(R_{s}^*(L_i+L_i^*)R_s\big)\Big)^2\,\dd s<\infty\]
for all $i\in I_b$ and
\[\int_0^{\infty}\eech\bigg(\Big(\frac{R_{s}^*C_j^*C_jR_{s}}{\tr(R_s^*C_j^*C_jR_s)}-M_{s}\Big)^2\tr(R_{s}^*C_j^*C_jR_s)\bigg)\,\dd s<\infty\]
for all $j\in I_p$. Since the integrands are nonnegative, their inferior limits at infinity are $0$. Hence there exists an unbounded increasing sequence $(t_n)_n$ such that for any $i\in I_b$,
\[\lim_n \eech\Big( R_{t_n}^*(L_i+L_i^*)R_{t_n}-M_{t_n}\tr\big(R_{t_n}^*(L_i+L_i^*)R_{t_n}\big)\Big)=0\]
and for any $j\in I_p$,
\[\lim_n \eech\bigg(\Big(\frac{R_{t_n}^*C_j^*C_jR_{t_n}}{\tr(R_{t_n}^*C_j^*C_jR_{t_n})}-M_{t_n}\Big)^2\tr(R_{t_n}^*C_j^*C_jR_{t_n})\bigg)=0.\]
Since convergence in $\Lp 1$ implies the almost sure convergence of a subsequence, there exists an unbounded increasing sequence, which we denote also by $(t_n)_n$, such that $\pp\ch$-almost surely,
\[\lim_{n\to\infty} \Big( R_{t_n}^*(L_i+L_i^*)R_{t_n}-M_{t_n}\tr(R_{t_n}^*(L_i+L_i^*)R_{t_n}) \Big)=0\]
and
\[\lim_{n\to\infty} \bigg(\Big(\frac{R_{t_n}^*C_j^*C_jR_{t_n}}{\tr(R_{t_n}^*C_j^*C_jR_{t_n})}-M_{t_n}\Big)^2\tr(R_{t_n}^*C_j^*C_jR_{t_n})\bigg)=0\]
for all $i\in I_b$ and $j\in I_p$.

Now and for the rest of this paragraph, fix a realization (i.e.\ an element of $\Omega$) such that $(M_{t_n})_n$ converges to $M_\infty$. The polar decomposition of $R_t$ is $R_t=U_t\sqrt{M_t}$. Since the set of $k\times k$ unitary matrices is compact, there exists a subsequence $(s_n)_n$ of $(t_n)_n$ such that $(U_{s_n})_n$ converges to $U_{\infty}$. We therefore have
\[\sqrt{M_{\infty}}U_\infty^*(L_i+L_i^*)U_\infty \sqrt{M_\infty}-M_\infty\tr(M_{\infty}U_\infty^*(L_i+L_i^*)U_\infty)=0\]
and
\[\sqrt{M_{\infty}}U_\infty^*C_j^*C_jU_\infty \sqrt{M_\infty}-M_\infty\tr(M_{\infty}U_\infty^*C_j^*C_jU_\infty )=0,\]
for all $i\in I_b$ and $j\in I_p$. Denoting $P_\infty$ the orthogonal projector onto the range of $M_\infty$, it follows that there exist real numbers $(\alpha_i)_{i\in I_b}$ and $(\beta_j)_{j\in I_p}$ such that
\[U_\infty P_\infty U_\infty^* (L_i+L_i^*) U_\infty P_\infty U_\infty^*=\alpha_i U_\infty P_\infty U_\infty^*\]
and
\[U_\infty P_\infty U_\infty^* C_j^*C_j U_\infty P_\infty U_\infty^*=\beta_j U_\infty P_\infty U_\infty^*.\]
Assumption {\pur} implies that the orthogonal projector $U_\infty P_\infty U_\infty^*$ has rank one, thus so does $P_\infty$. Since $\tr(M_\infty)=1$, $M_\infty$ is a rank one orthogonal projector.

Since $(M_{t_n})_n$ converges $\pp\ch$-almost surely, the above paragraph and the absolute continuity of $\pprho\rho$ with respect to $\pp\ch$ show that $M_\infty$ is $\pprho\rho$-almost surely a rank one orthogonal projector.
\end{proof}

\section{Invariant measure and exponential convergence in Wasserstein distance}\label{sec:inv}

This section is devoted to the main result of the paper, which concerns the exponential convergence to the invariant measure for the Markov process $(\hat x_t)_t$. We first show a convergence result for ${\cF}$-measurable random variables. The following theorem is a transcription of \cite[Proposition 7.5]{wolftour}.

\begin{thm}\label{thm:convavstate}
Assume that {\lerg} holds. Then there exist two constants $C>0$ and $\lambda>0$ such that for any $\rho\in \cD_k$ and any $t\geq0$,
\[\left\Vert \e^{t\Lin}(\rho)-\rhoinv\right\Vert\leq C\e^{-\lambda t}\]
\end{thm}

Our next proposition requires the introduction of a shift semigroup. From now on we assume that $\filtreP$ is a canonical realization of the processes $W_i$ and $N_j$, in particular $\Omega$ is a subset of $(\R^{I_b\cup I_p})^{\R_+}$. We can then define for every $t\geq 0$ the map $\theta^t$ on $\Omega$ by 
\[\big(\theta^t\omega\big)(s)=\omega(s+t)-\omega(t).\]
From the previous theorem we deduce the following proposition for $\cF$-measurable random variables.
\begin{prop}\label{prop:convergo}
Assume {\lerg} holds. Then there exist two constants $C>0$ and $\lambda>0$ such that for any $\cF$-measurable, essentially bounded function $f:\Omega\mapsto\C$ with essential bound $\|f\|_{\infty}$, any $t\geq 0$ and any $\rho\in\cD_k$,
\begin{equation}
\big\vert \eerho\rho(f\circ \theta^t)-\eerho\rhoinv(f)\big\vert\leq\Vert f\Vert_{\infty} C\e^{-\lambda t}.
\end{equation}
\end{prop}

\begin{proof}
Recall that by definition $\pp$ is the law of processes with independent increments. It follows that if $g$ is $\cF_t$-measurable and $h$ is $\cF$-measurable, $\ee(h\circ\theta^t g)=\ee(h\circ\theta^t)\ee(g)$. Then, by definition of $\pprho\rho$,
\begin{align*}
\eerho\rho(f\circ\theta^t)=\ee(f\circ\theta^tZ_{t+s}^{\rho}).
\end{align*}
Since $Z_{t+s}^{\rho}=\tr(S_{t+s}^{t\,*}S_{t+s}^t S_t\rho S_t^*)$ where $S_{t}^*\rho S_t$ is $\cF_t$-measurable and $S_{t+s}^{t\,*}S_{t+s}^t=S_{s}^*S_{s}\circ\theta^t$ by \eqref{eq_defS}
\begin{align*}
\eerho\rho(f\circ\theta^t)=\ee\Big(f\circ\theta^t\tr\big(S_{t+s}^{t\,*}S_{t+s}^t \ee(S_t\rho S_t^*)\big)\Big).
\end{align*}
Then relation \eqref{exppsrhos}, the $\theta$-invariance of $\pp$, and the definition of the measures $\pprho\rho$ yield
\begin{equation*}
\eerho\rho(f\circ\theta^t)=\eerho{\bar \rho_t}(f)
\end{equation*}
with $\bar\rho_t=\e^{t\Lin}(\rho)$.
It follows from Theorem \ref{thm:M_t} that 
\begin{equation*}
\eerho\rho(f\circ\theta^t)-\eerho{\rhoinv}(f)=\eech\Big( f \tr\big(M_\infty (\e^{t\Lin}(\rho )-\rhoinv)\big)\Big).
\end{equation*}
For any matrix $A$, denoting $\|A\|_1$ its trace norm, $\big|\tr(M_\infty A)\big|\leq\| A\|_1$. Therefore, 
\begin{equation*}
\big|\eerho\rho(f\circ\theta^t)-\eerho\rhoinv(f)\big|\leq \|f\|_\infty \|\e^{t\Lin}(\rho) -\rhoinv\|_1.
\end{equation*}
Theorem \ref{thm:convavstate} then yields the proposition.
\end{proof}

The main strategy to show Theorem \ref{thm:expo_conv_wasser1} is to construct a $\cF$-measurable process $(\hat y_t)_t$ approximating the process $(\hat x_t)_t$. Let $(\hat z_t)_t$ be the maximum likelihood process:
\begin{equation} \label{eq_defyn}
\hat z_{t}=\mathop{\mathrm{argmax}}_{\hat x\in \pc k}\,\|S_t x\|
\end{equation}
where $x$ is a norm one representative of $\hat x$. If the largest eigenvalue of $S_t^*S_t$ is not simple, the choice of $\hat z_t$ may not be unique.  However we can always choose an appropriate $\hat z_t$ in an $(\cF_t)_t$-adapted way. If {\pur} holds, Proposition \ref{prop:pur} ensures that the definition of $\hat z_t$ is almost surely unambiguous for large enough $t$: it is the equivalence class of eigenvectors of $M_t$ corresponding to its largest eigenvalue. 

Let now $(\hat y_t)_t$ be the evolution of this maximum likelihood estimate:
\begin{equation} \label{eq_defzn}
\hat y_t=S_t\cdot\hat z_t.
\end{equation}
We shall also use the notation $\hat z_{t+s}^s:=\hat z_{t}\circ\theta^s$ and $\hat y_{t+s}^s:=\hat y_t\circ\theta^s$, that is, processes defined in the same fashion but substituting $S_{t+s}^s$ for $S_t$. It is worth noticing that these processes are all $\cF$-measurable.

Our proof that $(\hat y_t)_t$ is an exponentially good approximation of $(\hat x_t)_t$ relies in part on the use of the exterior product of $\C^k$. We recall briefly the relevant definitions: for $x_1, x_2\in\C^k$ we denote by $x_1\wedge x_2$ the alternating bilinear form 
$$x_1\wedge x_2:(y_1,y_2)\mapsto \det\begin{pmatrix}
\langle x_1, y_1\rangle & \langle x_1, y_2\rangle\\ \langle x_2, y_1\rangle&\langle x_2, y_2\rangle
\end{pmatrix}.$$ 
Then, the set of all $x_1\wedge x_2$ is a generating family for the set $\wedge^2\C^k$ of alternating bilinear forms on $\C^k$. We equip it with a complex inner product by 
\[\langle x_1\wedge x_2, y_1\wedge y_2\rangle = \det\begin{pmatrix}
\langle x_1, y_1\rangle & \langle x_1, y_2\rangle\\ \langle x_2, y_1\rangle&\langle x_2, y_2\rangle
\end{pmatrix}, \]
and denote by $\|x_1\wedge x_2\|$ the associated norm (there should be no confusion with the norm on vectors). It is immediate to verify that our metric $d(\cdot,\cdot)$ on $\pc k$ satisfies
\begin{equation}\label{eq:d_wedge}
d(\hat x,\hat y)=\frac{\|x\wedge y\|}{\|x\|\|y\|}.
\end{equation}
For $A\in \mc k$, we write $\wedge^2 A$ for the operator on $\wedge^2\C^k$ defined by
\begin{equation}\label{eq_defwedgepA}
\big(\wedge^2 A\big) \,(x_1\wedge x_2)=Ax_1\wedge Ax_2.
\end{equation}
It follows that $\wedge^2 (AB)=\wedge^2 A\,\wedge^2 B$, so that $\|\wedge^2\! (AB)\|\leq\|\wedge^2 A\|\|\wedge^2 B\|$. There exists a useful relationship between the operator norm on $\wedge^2\mc k$ and singular values of matrices. From e.g.\  Chapter XVI of \cite{BMLalgebra},
\begin{equation}\label{eq:wedge_singular val}
\|\wedge^2 A\|=a_1(A)\,a_2(A),
\end{equation}
where $a_1(A)\geq a_2(A)$ are the two first singular values of $A$, i.e.\  the square roots of eigenvalues of~$A^* A$. We recall that the operator norm is defined such that $\|A\|:=a_1(A)$.

The exponential decrease of $d(\hat x_t,\hat y_t)$ is derived from the exponential decay of the following function:
\[f:t\mapsto\ee\big(\Vert\wedge^2S_t\Vert\big).\]

\begin{lem}\label{lem:f}
Assume that {\pur} holds. Then there exist two constants $C>0$ and $\lambda>0$ such that for all $t\geq0$
\[f(t)\leq C\e^{-\lambda t}\]
\end{lem}

\begin{proof}
First, we show that $f$ converges to zero as $t$ grows to $\infty$. To this end recall that $R_t=S_t/\sqrt{kZ_t\ch}$, so that
\begin{equation*}
\ee\big(\Vert\wedge^2 S_{t}\Vert\big)=\eech\big(k\Vert\wedge^2R_{t}\Vert\big).
\end{equation*}
Furthermore, since $R_t^*R_t=M_t$, we have from Theorem \ref{thm:M_t} and Proposition \ref{prop:pur} that
\[\lim_{t\rightarrow \infty}\Vert\wedge^2R_{t}\Vert=\lim_{t\rightarrow \infty}a_1(R_t)\,a_2(R_t)=0.\]
Indeed, since $a_1(R_t)$ and $a_2(R_t)$ are the largest two eigenvalues of $\sqrt{M_t}$, the fact that it converges to a rank-one projector implies that $a_1(R_t)$ converges to $1$ and $a_2(R_t)$ to zero. The inequality $\|S_t\|^2\leq \tr(S_t^*S_t)$ implies $\|\wedge^2 R_t\|\leq 1$ almost surely. Then Lebesgue's dominated convergence theorem yields $\lim_{t\to\infty}f(t)=0$.

Second, we show $f$ is submultiplicative. By the semi-group property, $S_{t+s}=S_{t+s}^sS_s$ for all $t,s\geq0$. Using that the norm is submultiplicative, for any $t,s\geq0$,
\[\Vert\wedge^2 S_{t+s}\Vert\leq\Vert\wedge^2 S_{t+s}^s\Vert\Vert\wedge^2 S_{s}\Vert\]
Since $\pp$ has independent increments, $\|\wedge^2 S_{t+s}^s\|=\|\wedge^2S_t\|\circ\theta^s$ is $\cF_s$-independent and $\|\wedge^2S_s\|$ is $\cF_s$-measurable,
\begin{equation*}
\ee \big(\Vert\wedge^2S_{t+s}\Vert\big) \leq \ee\big(\Vert\wedge^2 S_{t+s}^s\Vert\big) \ee\big(\Vert\wedge^2 S_{s}\Vert\big)
\end{equation*}
The measure $\pp$ being shift-invariant,
\begin{equation*}
\ee \big(\Vert\wedge^2S_{t+s}\Vert\big) \leq \ee\big(\Vert\wedge^2 S_{t}\Vert\big) \ee\big(\Vert\wedge^2 S_{s}\Vert\big)
\end{equation*}
which yields that $f$ is submultiplicative.

Since $f$ is measurable, submultiplicative and $0\leq f(t)\leq k$ for all $t$, Fekete's subadditive lemma ensures that there exists $\lambda\in (-\infty,\infty]$ such that
$$\lim_{t\rightarrow\infty}\frac{1}{t}\log f(t)=\inf \frac{1}{t}\log f(t)=-\lambda.$$
Since $f$ converges towards $0$, this $\lambda$ belongs to $(0,\infty]$. This yields the lemma.
\end{proof}

\begin{prop}\label{prop:convconv}
Assume that {\pur} holds. Then there exist two constants $C>0$ and $\lambda>0$ such that for any $s,t\in\R_+$ and for any probability measure $\mu$ on $(\pc k,\mathcal B)$,
\begin{equation}
\eemu\mu\big(d(\hat x_{t+s},\hat y_{t+s}^s)\big)\leq C\e^{-t\lambda}.
\end{equation}
\end{prop}

\begin{proof}
Recall that $\ee_\mu$ is the expectation with respect to $\qumu\mu$. Using the Markov property, we have
\begin{equation}
	\eemu\mu\big(d(\hat x_{t+s},\hat y_{t+s}^s)\big)=\mathbb E_{\mu_s}\big(d\big(\hat x_{t},\hat y_t)\big)
\end{equation}
with $\mu_s$ the distribution of $\hat x_s$ conditioned on $\hat x_0\sim \mu$.
Then it is sufficient to prove the proposition for $s=0$. For any $t\geq0$, using the fact that $\|S_tz_t\|=\|S_t\|$ for $z_t$ a norm one representative of $\hat z_t$,
\begin{align*}
d(\hat x_{t},\hat y_t)&=\frac{\Vert x_t\wedge y_t\Vert}{\Vert x_t\Vert\Vert y_t\Vert}
=\frac{\Vert S_t x_0\wedge S_t z_t\Vert}{\Vert S_t x_0\Vert\Vert S_t z_t\Vert}
\leq\frac{\Vert\wedge^2 S_t\Vert}{\Vert S_t  x_0\Vert\Vert S_t\Vert}
\leq \frac{\Vert\wedge^2 S_t\Vert}{\Vert S_t  x_0\Vert^2}
\end{align*}
Using this inequality and the fact that $\dd\qumu\mu{}|_{\cF_t\otimes\cB}=\|S_t x_0\|^2\,\dd\pp\,\dd\mu(\hat x_0)$,
\begin{equation*}
\eemu\mu\big(d(\hat x_{t},\hat y_t)\big)\leq\int \ee\big(\Vert\wedge^2 S_t\Vert\big)\, \dd \mu(\hat x_0)\leq f(t).
\end{equation*}
Finally Lemma \ref{lem:f} yields the proposition.
\end{proof}

We turn to the proof of our main theorem, Theorem \ref{thm:expo_conv_wasser1}. The speed of convergence is expressed in terms of the Wasserstein distance $W_1$. Let us recall the definition of this distance for compact metric spaces: for $X$ a compact metric space equipped with its Borel $\sigma$-algebra, the Wasserstein distance of order $1$ between two probability measures $\sigma$ and $\tau$ on $X$ can be defined using the Kantorovich--Rubinstein duality Theorem as
\[W_1(\sigma,\tau)=\sup_{f\in \mathrm{Lip}_1(X)}\Big| \int_{X} f\,\dd\sigma- \int _X f\, \dd\tau\Big|,\]
where $\mathrm{Lip}_1(X)=\{f:X\rightarrow\mathbb R \ \mathrm{s.t.}\ \vert f(x)-f(y)\vert\leq d(x,y)\}$ is the set of Lipschitz continuous functions with constant one, and $d$ is the metric on $X$. Here we use this for $X=\pc k$ and $d$ defined in \eqref{eq_defdistance} (see also \eqref{eq:d_wedge}).
\smallskip

We recall our main theorem before proving it.
\begin{thm}\label{thm:expo_conv_wasser}
Assume that {\pur} and {\lerg} hold. Then the Markov process $(\hat x_t)_t$ has a unique invariant probability measure $\mu\inv$, and there exist $C>0$ and $\lambda>0$ such that for any initial distribution $\mu$ of $\hat x_0$ over $\pc k$, for all $t\geq 0$, the distribution $\mu_t$ of $\hat x_t$ satisfies
\begin{equation*}
W_1(\mu_t,\mu\inv)\leq C\e^{-\lambda t}
\end{equation*}
where $W_1$ is the Wasserstein distance of order $1$.
\end{thm}

\begin{proof}
Let $f\in \mathrm{Lip}_1(\pc k)$. From the definition of Wasserstein distance, we can restrict ourselves to functions $f$ that vanish at some point. Remark that since $\sup_{\hat x,\hat y\in \pc k}d(\hat x,\hat y)=1$, restricting to this set of functions implies $\|f\|_\infty\leq 1$. Let $\mu\inv$ be an invariant probability measure for $(\hat x_t)_t$. We will prove the exponential convergence of $(\mu_t)_t$ towards $\mu\inv$ for any initial $\mu_0$, and that will imply that $(\hat x_t)_t$ accepts a unique invariant probability measure. Let $t\geq 0$, and recall that $\hat y_{t}^{t/2} = \hat y_{t/2} \circ \theta^{t/2}$. We have
\begin{align}
\eemu\mu\big(f(\hat x_t)\big)-\eemu{\mu\inv}\big(f(\hat x_t)\big)&=\eemu\mu\big(f(\hat x_t)\big)-\eemu\mu\big(f(\hat y_t^{t/2})\big)+ \eemu{\mu\inv}\big(f(\hat y_t^{t/2})\big)-\eemu{\mu\inv}\big(f(\hat x_t)\big)  \nonumber\\
&\qquad + \eemu\mu\big(f(\hat y_t^{t/2})\big)-\eemu{\mu\inv}\big(f(\hat y_t^{t/2})\big)\nonumber\\&\leq\eemu\mu \big(d(\hat x_t,\hat y_t^{t/2})\big)+ \eemu{\mu\inv}\big(d(\hat x_t,\hat y_t^{t/2})\big)\label{eq:line_approx}  \\
&\qquad + \eemu\mu\big(f(\hat y_t^{t/2})\big)-\eemu{\mu\inv}\big(f(\hat y_t^{t/2})\big)\label{eq:line_erg}.
\end{align}
The two terms on the right hand side of line \eqref{eq:line_approx} are bounded using Proposition \ref{prop:convconv}.  Using Proposition \ref{prop:margin}, the difference on line \eqref{eq:line_erg} satisfies\[\eemu\mu\big(f(\hat y_t^{t/2})\big)-\eemu{\mu\inv}\big(f(\hat y_t^{t/2})\big)= \eerho{\rho_\mu}\big(f(\hat y_t^{t/2})\big)-\eerho{\rho\inv}\big(f(\hat y_t^{t/2})\big).\]
Then bounding the right hand side using Proposition \ref{prop:convergo}, it follows there exist $C>0$ and $\lambda>0$ such that
\[
\Big|\eemu\mu\big(f(\hat x_t)\big)-\eemu{\mu\inv}\big(f(\hat x_t)\big)\Big|\leq3C \e^{-\lambda t/2 }.
\]
Adapting the two constants yields the theorem.
\end{proof}

\section{Set of invariant measures under {\bf(Pur)}}\label{app:inv_measure}
The results and proofs of this section are a direct translation of \cite[Appendix B]{Inv1}. We reproduce the proofs for the reader's convenience.

Whenever {\lerg} does not hold, $\dim\ker \Lin>1$ and the semigroup $(\e^{t \Lin})_t$ accepts more than one fixed point in $\cD_k$. The convex set of invariant states can be explicitly classified given the matrices $(L_i)_{i\in I_b}$ and $(C_j)_{j\in I_b}$. Following \cite[Theorem 7]{Baumgartner} (alternatively see Theorem~7.2 and Proposition~7.6 in \cite{wolftour}, and \cite{ticozziviola}), there exists a decomposition
\[
\mathbb{C}^k  \simeq \mathbb{C}^{n_1} \oplus\cdots\oplus \mathbb{C}^{n_d} \oplus \mathbb{C}^{D}, \quad k = n_1 + \ldots + n_d + D
\]
with the following properties: 
\begin{enumerate}
\item The range of any invariant states is a subspace of $V = \mathbb{C}^{n_1} \oplus\cdots\oplus \mathbb{C}^{n_d} \oplus \{0\}$;
\item\label{it:e2} The restriction of the operators $L_i$ and $C_j$ to $ \mathbb{C}^{n_1} \oplus\cdots\oplus \mathbb{C}^{n_d}$ are block-diagonal, with
\begin{equation}
\label{eq:e2}
\begin{aligned}
L_i &= L_{1,i} \oplus\cdots\oplus L_{d,i},&& i\in I_b,\\
C_j &= C_{1,j} \oplus\cdots\oplus C_{d,j},&& j\in I_p;
\end{aligned}
\end{equation}
\item\label{it:e3} For each  $\ell=1,\ldots,d$ there is a decomposition $\mathbb{C}^{n_\ell} = \mathbb{C}^{k_\ell} \otimes \mathbb{C}^{m_\ell}, \, n_\ell = k_\ell \times m_\ell$, a unitary matrix $U_\ell$ on $\mathbb{C}^{n_\ell}$  and matrices $\{\hat L_{\ell,i}\}_{i\in I_b}$ and $\{\hat C_{\ell,j}\}_{j\in I_p}$ on $\mathbb{C}^{k_\ell}$ such that
\begin{equation}
\label{eq:e3}
\begin{aligned}
L_{\ell,i}  &= U_\ell (\hat L_{\ell,i} \otimes \Id_{\C^{m_\ell}}) U_\ell^*,&& i\in I_b,\\
C_{\ell,j} &= U_\ell (\hat C_{\ell,j} \otimes \Id_{\C^{m_\ell}}) U_\ell^*,&& j\in I_p;
\end{aligned}
\end{equation}
\item\label{it:e4} There exists a positive definite matrix $\rho_\ell$ on $\mathbb{C}^{k_\ell}$ such that
\begin{equation}
0 \oplus\cdots\oplus U_\ell (\rho_\ell \otimes \Id_{\C^{m_\ell}}) U_\ell^* \oplus\cdots\oplus 0 
\end{equation}
is a fixed point of $(\e^{t\Lin})_t$.
\end{enumerate}

Then, the set of fixed points for $(\e^{t\Lin})$ is
\begin{equation*}
U_1\big(\rho_1\otimes M_{m_1}(\C)\big)U_1^*\oplus\ldots \oplus U_d\big(\rho_d\otimes M_{m_d}(\C)\big)U_d^*\oplus0_{M_D(\C)}.
\end{equation*}
The decomposition simplifies under the purification assumption.
\begin{prop}\label{lem:phi_FP}
Assume that {\pur} holds. Then there exists a set $\{\rho_\ell\}_{\ell=1}^d$ of positive definite matrices and an integer $D$ such that the set of fixed points of $(\e^{t\Lin})_t$ is
\[\C\rho_1\oplus\cdots\oplus\C\rho_d\oplus 0_{M_D(\C)}.\]
\end{prop}
\begin{proof}
The statement follows from the discussion preceding the proposition if we show that {\pur} implies $m_1 = \ldots = m_d =1$. Assume that one of the $m_\ell$, e.g.\ $m_1$, is greater than $1$. Let $x$ be a norm one vector in $\mathbb{C}^{k_1}$. Then $\pi = U_1(\pi_{\hat{x}} \otimes I_{\mathbb{C}^{m_1}}) U_1^*\oplus 0 \oplus\cdots\oplus0$ is an orthogonal projection of rank $m_1>1$, and
\[
 \pi (L_i+L_i^* )\pi  = \|(L_{1,i}+L_{1,i}^*) x\|^2\, \pi \quad\mbox{and}\quad \pi (C_j^*C_j )\pi  = \|C_{1,j}^*C_{1,j} x\|^2 \,\pi\quad  \mbox{for all }i\in I_b, j\in I_p,
\]
and this contradicts \pur.
\end{proof}

It is clear from Proposition \ref{eq:e2} that to each extremal fixed point $0 \oplus \cdots \oplus \rho_\ell \oplus\cdots\oplus0$ corresponds a unique invariant measure $\mu_\ell$ supported on its range $\ran \rho_\ell$. The converse is the subject of the next proposition.
\begin{prop}\label{prop:set_inv_measure}
Assume {\pur} holds. Then any invariant probability measure of $(\hat x_t)_t$ is a convex combination of the measures $\mu_\ell$, $\ell=1,\ldots,d$.
\end{prop}
\begin{proof}
Let $\mu$ be an invariant probability measure for $(\hat x_t)_t$ and $f$ be a continuous function on $\pc k$. Proposition \ref{prop:convconv} implies that
\[\int f \,\dd\mu=\eemu\mu\big(f(\hat x_0)\big)=\eemu\mu\big(f(\hat x_t)\big)=\lim_{t\to\infty}\eemu\mu\big(f( \hat y_t)\big).\]
Since $(\hat y_t)_t$ is $\cF$-measurable, Proposition~\ref{prop:margin} implies $\int f \,\dd\mu=\lim_{t\to\infty}\eerho{\rho_\mu}\big(f( \hat y_t)\big)$, and by Remark \ref{rmq_muinvrhoinv}, $\rho_\mu\in\cD_k$ is a fixed point of $(\e^{t\Lin})_t$. Proposition \ref{lem:phi_FP} ensures that there exist nonnegative numbers $p_1,\ldots,p_d$ summing up to one such that $\rho_\mu=p_1\,\rho_1\oplus\cdots\oplus p_d\,\rho_d\oplus 0_{M_D(\C)}$. From the definition of $\pprho{\rho_\mu}$,
\[\pp^{\rho_\mu}=p_1\,\pp^{\rho_1}+\cdots+p_d\,\pp^{\rho_d}\]
with the abuse of notation $\rho_\ell:= 0\oplus\cdots\oplus \rho_\ell\oplus\cdots\oplus 0$, so that
\[\int f \, \dd\mu=\lim_{t\to\infty} p_1\,\eerho{\rho_1}\big(f(\hat y_t)\big)+\cdots+p_d\,\eerho{\rho_d}\big(f(\hat y_t )\big).\]
The same argument gives $\int f \, \dd\mu_\ell= \lim_{t\to\infty} \eerho{\rho_\ell}\big(f(\hat y_t)\big)$, and we have $\mu=p_1\,\mu_1+\ldots+ p_d\,\mu_d$.
\end{proof}

\section{{\bf(Pur)} is not necessary for purification}\label{sec:pur_not_nec}
As shown by the following example, the condition {\pur} is sufficient but not necessary for \eqref{eq:purification} to hold. 

Let $k=3$ and fix an orthonormal basis $\{e_1,e_2,e_3\}$ of $\C^3$. Let $I_b=\{0,1\}$, $I_p=\{2\}$, $u=(e_1+e_2+e_3)/\sqrt{3}$ and $v=(e_1+e_3)/\sqrt{2}$. Let
\begin{equation} \label{eq_defcontrexpur}
	H=0,\qquad V_0=L_0=e_1u^*, \qquad V_1=L_1=2 vv^*+e_2e_2^*, \qquad V_2=C_2=u e_1^*.
\end{equation}
\begin{prop}\label{prop:Lirr-pur}
	Let $\Lin$ be the Lindblad operator given by \eqref{eq_deflindblad} with $H$, $V_1$, $V_2$, $V_3$ defined in~\eqref{eq_defcontrexpur}. Then {\lerg} holds and the unique invariant state $\rhoinv$ is positive definite.
\end{prop}
\begin{proof}
	Using \cite[Proposition 7.6]{wolftour}, it is sufficient to prove that if $\pi$ is a non null orthogonal projector such that $(\Id-\pi) L_0\pi=(\Id-\pi) L_1\pi = (\Id-\pi)C_2\pi=0$, then $\pi=\Id$. Assume $\operatorname{rank}\pi<3$. Since $\pi\in M_3(\C)$, there exist $\hat x\in \pc 3$ such that either $\pi=\pi_{\hat x}$ or $\pi=\Id-\pi_{\hat x}$. If the first alternative holds, $(\Id-\pi) L_0\pi=(\Id-\pi) L_1\pi = (\Id-\pi)C_2\pi=0$ implies $\hat x$ is the equivalence class of a common eigenvector of $L_0$, $L_1$ and $C_2$. If the second alternative holds, $\hat x$ is the equivalence class of a common eigenvector of $L_0^*$, $L_1^*$ and $C_2^*$. The only common eigenvectors of $L_0$ and $C_2$ or $L_0^*$ and $C_2^*$ are elements of $\C (e_2-e_3)$. Since $L_1$ is self adjoint, and this eigenspace is not an eigenspace of $L_1$, the proposition holds.
\end{proof}

 In the orthonormal basis $\{e_1,e_2,e_3\}$,
$$L_0^*+L_0=\frac{1}{\sqrt{3}}\begin{pmatrix}
2&1&1\\ 1&0&0\\ 1&0&0
\end{pmatrix},\quad L_1^*+L_1=2\begin{pmatrix}
1&0&1\\0&1&0\\1&0&1
\end{pmatrix}
\and C_2^*C_2=\begin{pmatrix}
1&0&0\\0&0&0\\0&0&0
\end{pmatrix}.$$
Taking $\pi$ the orthogonal projector onto the subspace spanned by $\{e_2,e_3\}$ it follows that {\pur} does not hold.
Yet we have the following proposition.

\begin{prop} \label{prop:purcontrex}
Consider the family of processes $(\rho_t)_t$ defined by \eqref{eq1:SME} with $H$, $L_0$, $L_1$, $C_2$ defined in \eqref{eq_defcontrexpur}. Then for any $\rho\in \cD_k$,
\[\lim_{t\to\infty} \inf_{\hat y\in \pc k}\|\rho_t -\pi_{\hat y}\|=0\quad \pp^{\rho}\text{-almost surely.}\]
\end{prop}
\begin{proof}
Proposition \ref{prop:Lirr-pur} implies that $\rhoinv$, the unique element of $\cD_k$ invariant by $(\e^{t\Lin})_t$ is positive definite. Then, $\tr(C_2^*C_2\rhoinv)>0$. The results of \cite{KuMa} thus ensure that for any $\rho\in \cD_k$,
\[\lim_{t\to \infty} N_2(t)/t=\tr(C_2^*C_2\,\rhoinv),\quad\pp^{\rho}\text{-almost surely.}\]
Let $T=\inf\{t\geq 0 : N_2(t)\geq 1\}$. Then $\pp^{\rho}(T<\infty)=1$ and from the definition of $C_2$,
\[\rho_T=\pi_{\hat u}\and \rho_t=\pi_{S_t^T\cdot \hat u} \text{ for any }t\geq T.\]
Hence $\inf_{\hat y\in \pc k}\|\rho_t -\pi_{\hat y}\|=0$ for any $t\geq T$ and $\pp^{\rho}(T<\infty)=1$ yield the proposition.
\end{proof}

\begin{cor}
Consider the process $(\hat x_t)_t$ defined by \eqref{eq:SDEpur} with $H$, $L_0$, $L_1$, $C_2$ defined in \eqref{eq_defcontrexpur}. Then $(\hat x_t)_t$ accepts a unique invariant probability measure $\mu\inv$ and there exist $C>0$ and $\lambda>0$ such that for any initial distribution $\mu$ of $\hat x_0$ over $\pc 3$, for all $t\geq 0$, the distribution $\mu_t$ of $\hat x_t$ satisfies
\[W_1(\mu_t,\mu\inv)\leq C e^{-\lambda t}.\]
\end{cor}
\begin{proof}
It is a direct adaptation of our proof of Theorem \ref{thm:expo_conv_wasser1}. Indeed Theorem \ref{thm:expo_conv_wasser1} holds if one substitutes the conclusion of Proposition \ref{prop:pur} for \pur. Taking $\rho_0=\Id/3$ in the latter proposition yields $\rho_t=\frac{S_t S_t^*}{\tr(S_t S_t^*)}$ and $M_t=\frac{S_t^* S_t}{\tr(S_t S_t^*)}$. Therefore, $\rho_t$ and $M_t$ are unitarily equivalent. Following the arguments and notation of the proof of Proposition \ref{prop:purcontrex} we see that $\pp\ch$-almost surely, $M_T$ has rank one, and so does any $M_t$ for $t\geq T$. Hence the conclusion of Proposition \ref{prop:pur} holds and the corollary is proven.
\end{proof}

Following the proofs of the discrete-time results of \cite{Inv1}, we can prove that the implication in Proposition \ref{prop:pur} is an equivalence if {\pur} is replaced by
\hypertarget{nscpur}{}
\begin{description}
\item[\nscpur] Any non zero orthogonal projector $\pi$ that satisfies $\pi S_t^*S_t\pi\propto\pi$ $\pp$-almost surely for any $t\geq 0$ has rank one.
\end{description}

\noindent Alas, in practice, such a condition is hard to check.

\section{Examples}\label{sec:ex}
In the following examples $k=2$. We recall the definition of the Pauli matrices:
\begin{align*}
\sigma_x:=\begin{pmatrix}
0&\phantom{-}1\\1&\phantom{-}0
\end{pmatrix},\quad
\sigma_y:=\begin{pmatrix}
0&-\ii\\\ii&\phantom{-}0
\end{pmatrix}\quad\mbox{and}\quad
\sigma_z:=\begin{pmatrix}
1&\phantom{-}0\\0&-1
\end{pmatrix}.
\end{align*}
A standard orthonormal basis of $M_2(\C)$ equipped with the Hilbert--Schmidt inner product is 
\[\left(\tfrac{1}{\sqrt{2}}\Id,\tfrac{1}{\sqrt{2}}\sigma_x,\tfrac{1}{\sqrt{2}}\sigma_y,\tfrac{1}{\sqrt{2}}\sigma_z\right).\]
In the basis of Pauli matrices one can write in a unique way any projection $\pi_{\hat x}$ as
\[ \pi_{\hat x}=\tfrac12\big(\Id + \x \sigma_x + \y \sigma_y+\z \sigma_z\big)\]
where
\[\x= \tr (\pi_{\hat x}\sigma_x), \quad\y= \tr (\pi_{\hat x}\sigma_y), \quad\z= \tr (\pi_{\hat x}\sigma_z).\]
We denote in particular by $\x_t,\y_t,\z_t$ respectively the coordinates associated with $\pi_{\hat x_t}$.

\subsection{Unitarily perturbed non demolition diffusive measurement}\label{subsec_UPQND}
Our first example consists of a $\tfrac12$-spin (or qbit) in a magnetic field oriented along the $y$-axis and subject to indirect non demolition measurement along the $z$-axis. It is a typical quantum optics experimental situation (see for example \cite{ficheux2018dynamics}). In terms of the parameters defining the related quantum trajectories, we get $H=\sigma_y$, $I_b=\{0\}$, $I_p=\emptyset$ and $L_0=\sqrt{\gamma}\,\sigma_z$ with $\gamma>0$. Then $(\pi_{\hat x_t})_t$ conditioned on $\hat x_0$ is the solution of 
\begin{equation}\label{eq:SDE_unitary_QND}
\dd\pi_{\hat x_t}=\big(-\ii[\sigma_y,\pi_{\hat x_t}] +\gamma(\sigma_z\pi_{\hat x_t}\sigma_z-\pi_{\hat x_t})\big)\,\dd t + \sqrt{\gamma}\big(\sigma_z\pi_{\hat x_t}+\pi_{\hat x_t}\sigma_z-2\tr(\sigma_z\pi_{\hat x_t})\pi_{\hat x_t}\big)\,\dd B_t.
\end{equation}
For this quantum trajectory it is immediate to verify {\pur}, and solving $\Lin(\rho)=0$ shows that $\rho\inv=\tfrac12\Id$ is the unique invariant state, so that {\lerg} holds. Hence, by Theorem \ref{thm:expo_conv_wasser1} $(\hat x_t)_t$ has a unique invariant measure. In the following we derive an explicit expression for this invariant measure.

The next Lemma allows us to restrict the state space.
\begin{lem}\label{lem:unitary_QND_delta_phi}
If $\mu(\y_0=0)=1$ then $\qumu\mu(\y_t=0)=1$ for all $t$ in $\R$.
\end{lem}
\begin{proof}
From equation \eqref{eq:SDE_unitary_QND}, $(\y_t)_t$ is the solution of
$\dd \y_t=2\y_t(\gamma \,\dd t-\sqrt{\gamma} \,\z_t\,\dd B_t)$.
It is therefore a Doléans-Dade exponential:
\[\y_t=\y_0\ e^{2\gamma t}	\exp\big(-2\gamma\int_0^t \z_s^2\, \dd s -2\sqrt{\gamma}\int_0^t \z_s \, \dd B_s\big)\]
and the conclusion follows.
\end{proof}

Now we prove that the invariant measure admits a rotational symmetry.
\begin{lem}\label{lem:symmetry_sigma_y} 
	Assume that the distribution $\mu$ of $\hat x_0$ is invariant with respect to the mapping $\hat x\mapsto \sigma_y\cdot\hat x$. Then $\mu_t$ is invariant with respect to the same mapping.
\end{lem}
\begin{proof}
Since $\sigma_y$ is unitary and self-adjoint, we have $\pi_{\sigma_y\cdot\hat x}=\sigma_y \pi_{\hat x}\sigma_y$. Since $\sigma_y\sigma_z=-\sigma_z\sigma_y$, it follows from \eqref{eq:SDE_unitary_QND} that
\begin{align*}
\dd(\sigma_y\pi_{\hat x_t}\sigma_y)=&\big(-\ii[\sigma_y,\sigma_y\pi_{\hat x_t}\sigma_y]+\gamma(\sigma_z\sigma_y\pi_{\hat x_t}\sigma_y\sigma_z- \sigma_y\pi_{\hat x_t}\sigma_y)\big)\,\dd t\\
		& -\sqrt{\gamma}\big(\sigma_z\sigma_y\pi_{\hat x_t}\sigma_y+\sigma_y\pi_{\hat x_t}\sigma_y\sigma_z-2\tr(\sigma_z\sigma_y\pi_{\hat x_t}\sigma_y)\sigma_y\pi_{\hat x_t}\sigma_y\big)\,\dd B_t.
\end{align*}
Then it follows from $\sigma_y \cdot \hat x_0\sim\hat x_0$ and $(B_t)_t\sim (-B_t)_t$ that $(\hat x_t)_t$ and $(\sigma_y\cdot\hat x_t)_t$ are both weak solutions to the same SDE with the same initial condition. Since this SDE has a unique solution, they have the same distributions.
\end{proof}

\begin{prop}\label{prop:unitary_inv_meas}
Let $(\hat x_t)_t$ be the process defined by \eqref{eq:SDE_unitary_QND}. Then its unique invariant measure is the normalized image measure by
\begin{equation}\label{eq_muinvimagepar}
	\iota:\theta\mapsto \tfrac12\big(\Id + \sin\theta\,\sigma_x +\cos\theta \,\sigma_z\big)
\end{equation}
of the measure $\tau(\theta)\,\dd\theta$ on $(-\pi,\pi]$ with
\[\tau(\theta)=\int_\theta^\pi \exp\frac{\cot x-\cot \theta}{\gamma} \frac{\sin x}{\sin^3\theta}\,\dd x\]
for $\theta\in[0,\pi]$ and $\tau(\theta)=\tau(\theta+\pi)$ for $\theta\in(-\pi,0]$.
\end{prop}
\begin{proof} The convergence results in Theorem \ref{thm:expo_conv_wasser1} and Lemma \ref{lem:unitary_QND_delta_phi} imply that the invariant measure $\mu\inv$ is the image by $\iota$ of a probability measure $\uptau$ on $(-\pi,\pi]$. Let $(\theta_t)_t$ be the solution of
\begin{equation}\label{eq:SDE_theta_unitary_QND}
\dd\theta_t=2(1-\gamma \cos\theta_t\sin\theta_t)\,\dd t -2\sqrt{\gamma}\sin\theta_t \,\dd B_t
\end{equation}
with initial condition $\theta_0$. Remark that $(\theta_t)_t$ is $2\pi$-periodic with respect to its initial condition, namely, $(\theta_t+2\pi)_t$ is solution of \eqref{eq:SDE_theta_unitary_QND} with initial condition $\theta_0+2\pi$. Now, using the It\^o formula,
\[(\cos\theta_t,\sin\theta_t)_t\sim\big(\tr(\pi_{\hat x_t}\sigma_z),\tr(\pi_{\hat x_t}\sigma_x)\big)_t\]
for $(\pi_{\hat x_t})_t$ solution of \eqref{eq:SDE_unitary_QND} with initial condition $\hat x_0=\tfrac12\big(\Id + \sin\theta_0\,\sigma_x +\cos\theta_0 \,\sigma_z\big)$.
Hence $\big(\iota(\theta_t)\big)_t$ has the same distribution as $(\hat x_t)_t$. Therefore $\uptau$ is an invariant measure for the diffusion defined by \eqref{eq:SDE_theta_unitary_QND}; in addition, Theorem \ref{thm:expo_conv_wasser1} shows that this invariant measure is unique, and Lemma \ref{lem:symmetry_sigma_y} shows that it is $\pi$-periodic. Following standard methods (see \cite{MR611513}), one shows that the restriction of $\uptau$ to $[0,\pi)$ has a density of the form $\tau(\theta)=C_1\tau_1(\theta)+C_2\tau_2(\theta)$ with $C_1,C_2\in\R$ and
\[\tau_1(\theta)=\frac{\int_\theta^\pi \sin x \, \exp (\tfrac1\gamma \cot  x) \, \dd x}{\sin^3\theta  \exp(\frac1\gamma \cot \theta)} \qquad \tau_2(\theta)=\frac1{ \sin^3\theta  \exp(\tfrac1\gamma \cot \theta)}\]
Now, straightforward analysis shows that $\int_0^\pi \tau_1(\theta)\,\dd\theta<\infty$ while $\int_0^\pi \tau_2(\theta)\,\dd\theta=\infty$. Therefore, $\tau$ is proportional to $\tau_1$ and the result follows.

\end{proof}
\begin{rmq}
	For $\gamma\to\infty$, the invariant measure $\mu\inv$ in Proposition \ref{prop:unitary_inv_meas} is a Dirac measure at $0$ and~$\pi$. To describe the scaling for $\gamma$ large we embed $\tau(\theta)$ into $L^{1}(\mathbb{R})$ by defining it to be zero outside the region $(-\pi, \pi]$. Then on the positive half line, in the $L^1$ norm,
	$$
	\lim_{\gamma \to \infty}\frac{1}{2 \gamma^3} \tau(\frac{\theta}{\gamma}) = \frac{1}{\theta^3} \exp(-\frac{1}{\theta}).
	$$
Hence, for large $\gamma$, the stationary probability distribution has two peaks of width (of order) $1/\gamma$ located $1/\gamma$ radians clockwise from the limit points $0$ and $\pi$. Furthermore the probability to find the particle around the limit points is exponentially suppressed.

The strong noise limit, $\gamma \to \infty$, was recently studied in various models \cite{BCFFS,BBT15,BCCNP}. This is the first model that allows for an explicit calculation of the shape of the stationary probability measure. The density of the invariant probability distribution is plotted in Figure~\ref{figure_UPQND} for three values of $\gamma$, and for $\theta \in [0,\pi]$.

\begin{figure}[ht]
\includegraphics[width=1.0\textwidth]{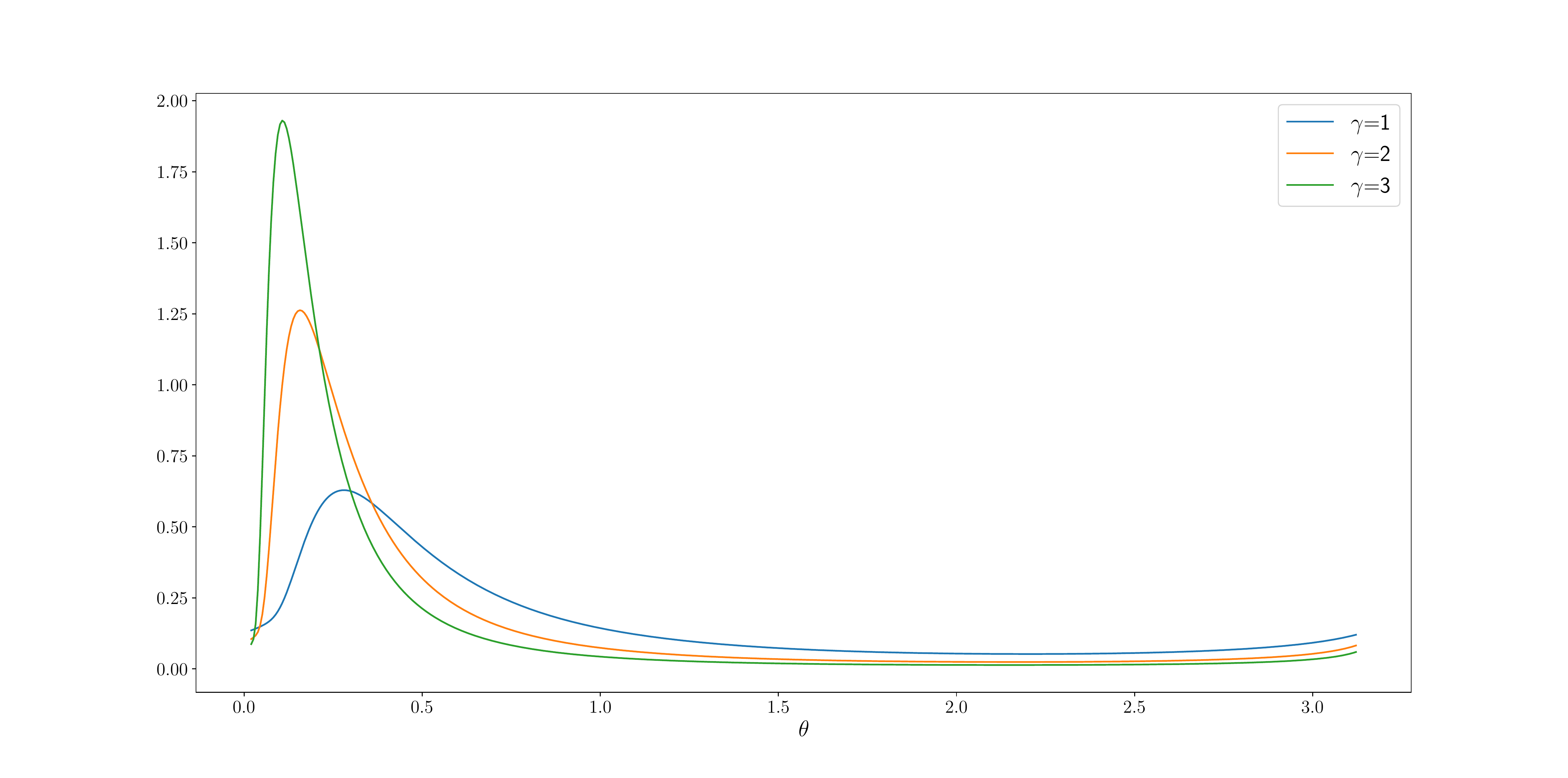}
\vspace{-2em}\caption{The restriction to $[0,\pi]$ of the density of the invariant probability distribution in Example \ref{subsec_UPQND}.}
\label{figure_UPQND}
\end{figure}
\end{rmq}

\subsection{Thermal qubit, diffusive case} \label{subsec_thqubitdiff}
The following second example corresponds to the evolution of a qubit interacting weakly with the electromagnetic field at a fixed temperature. The emission and absorption of photons by the qubit are stimulated by a resonant coherent field (laser). In the limit of a strong stimulating laser, the measurement of emitted photons results in a diffusive signal whose drift depends on the instantaneous average value of the raising and lowering operators of the qubit (see \cite[\S4.4]{wisemanmilburn} for a more detailed physical derivation). We obtain an analytically solvable model if we assume that the unitary rotation of the qubit is compensated for and thus frozen. In terms of the parameters defining the related quantum trajectories, we get $H=0$, $I=I_b=\{0,1\}$, $L_0=\sqrt{a}  \,\sigma_+$ and $L_1=\sqrt{b}\,\sigma_-$ with $a,b\in\mathbb R_+\setminus\{0\}$ and $\sigma_\pm=\tfrac12(\sigma_x\pm \ii\sigma_y)$, so that $\sigma_+= \begin{pmatrix} 0 & 1 \\ 0 & 0\end{pmatrix}$ and $\sigma_-= \begin{pmatrix} 0 & 0 \\ 1 & 0\end{pmatrix}$.

The stochastic master equation satisfied by $\pi_{\hat x_t}$ is
\begin{equation}\label{eq:SDE11}
\begin{aligned}
\dd\pi_{\hat x_t}&=\quad a\Big(\sigma_+\pi_{\hat x_t}\sigma_--\tfrac{1}{4}\big((\Id-\sigma_z)\pi_{\hat x_t}+\pi_{\hat x_t}(\Id-\sigma_z)\big)\Big)\,\dd t\\
&\quad+b\Big(\sigma_-\pi_{\hat x_t}\sigma_+-\tfrac{1}{4}\big((\Id+\sigma_z)\pi_{\hat x_t}+\pi_{\hat x_t}(\Id+\sigma_z)\big)\Big)\,\dd t\\
&\quad+ \sqrt{a}\Big(\sigma_+\pi_{\hat x_t}+\pi_{\hat x_t}\sigma_--\tr(\sigma_x\pi_{\hat x_t})\pi_{\hat x_t}\Big)\,\dd B_0(t)\\
&\quad+\sqrt{b}\Big(\sigma_-\pi_{\hat x_t}+\pi_{\hat x_t}\sigma_+-\tr(\sigma_x\pi_{\hat x_t})\pi_{\hat x_t}\Big)\,	\dd B_1(t)
\end{aligned}
\end{equation}
Again it is immediate to verify {\pur}, and solving for $\Lin(\rho)=0$ shows that {\lerg} holds.

\begin{lem}\label{lem:therm_Y_0}
If $\mu(\y_0=0)=1$ then  $\qumu\mu(\y_t=0)=1$ for all $t$ in  $\R$.
\end{lem}
\begin{proof}
From \eqref{eq:SDE11}, $(\x_t)_t$ and $(\y_t)_t$ satisfy
\[\dd \y_t=-\y_t\Big(\tfrac12(a+b)\,\dd t+\x_t\big(\sqrt{a}\,\dd B_0(t)+\sqrt{b}\,\dd B_1(t)\big)\Big).\]
Therefore, if one defines
\[M_t=\exp\Big(-\frac12\int_0^t (a+b)\x_s^2\, \dd s -\int_0^t \x_s\big(\sqrt{a}\,\dd B_0(s)+\sqrt{b}\,\dd B_1(s)\big)\Big)\]
then one has $\y_t=\y_0 \,\e^{-\tfrac12(a+b)t}M_t$, and this proves Lemma \ref{lem:therm_Y_0}.
\end{proof}

\begin{prop}
Let $(\hat x_t)_t$ be the process defined by \eqref{eq:SDE11}. Then its unique invariant measure is the normalized image measure  by $\iota$ (defined by \eqref{eq_muinvimagepar}) of the measure $\tau(\theta) \, \dd\theta$ on $(-\pi,\pi]$ with
\[
\tau(\theta)=\frac{e^{\varsigma z\arctan\big(\varsigma(\cos\theta-z)\big)}}{\big(\cos^2\theta+1-2z\cos\theta)\big)^{3/2}},
\]
with $z=\frac{a-b}{a+b}$ and $\varsigma=\frac{a+b}{2\sqrt{ab}}.$
\end{prop}

\begin{proof}
As in the proof of Proposition \ref{prop:unitary_inv_meas}, Theorem \ref{thm:expo_conv_wasser1} and Lemma \ref{lem:therm_Y_0} imply that the invariant measure $\mu\inv$ is the image by $\iota$ of a probability measure $\uptau$ on $(-\pi,\pi]$. Let $(\theta_t)_t$ be the solution of 
\begin{align}\label{eq_eqthetathermalqubitdiff}
\dd \theta_t = \big((b-a)\sin\theta_t + \tfrac12(a+b) \cos\theta_t\sin\theta_t\big) \, \dd t + \sqrt a \,(\cos\theta_t-1)\,\dd B_0(t) + \sqrt b \,(\cos\theta_t +1)\,\dd B_1(t).
\end{align}
The It\^o formula implies once again
\[(\cos\theta_t,\sin\theta_t)\sim\big(\tr(\pi_{\hat x_t}\sigma_z),\tr(\pi_{\hat x_t}\sigma_x)\big)_t\]
for $(\pi_{\hat x_t})_t$ solution of \eqref{eq:SDE11} with initial condition $\hat x_0=\tfrac12\big(\Id + \sin\theta_0\,\sigma_x +\cos\theta_0 \,\sigma_z\big)$.
Hence $\big(\iota(\theta_t)\big)$ has the same distribution as $(\hat x_t)$. As in the proof of Proposition \ref{prop:unitary_inv_meas}, standard techniques show that the unique invariant distribution for \eqref{eq_eqthetathermalqubitdiff} has density proportional to the function $\tau$ above.
\end{proof}
The density of the invariant probability distribution for three values of the pair $(a,b)$ is plotted in Figure~\ref{figure_thqubitdiff}.
\begin{figure}[ht]
\includegraphics[width=1.0\textwidth]{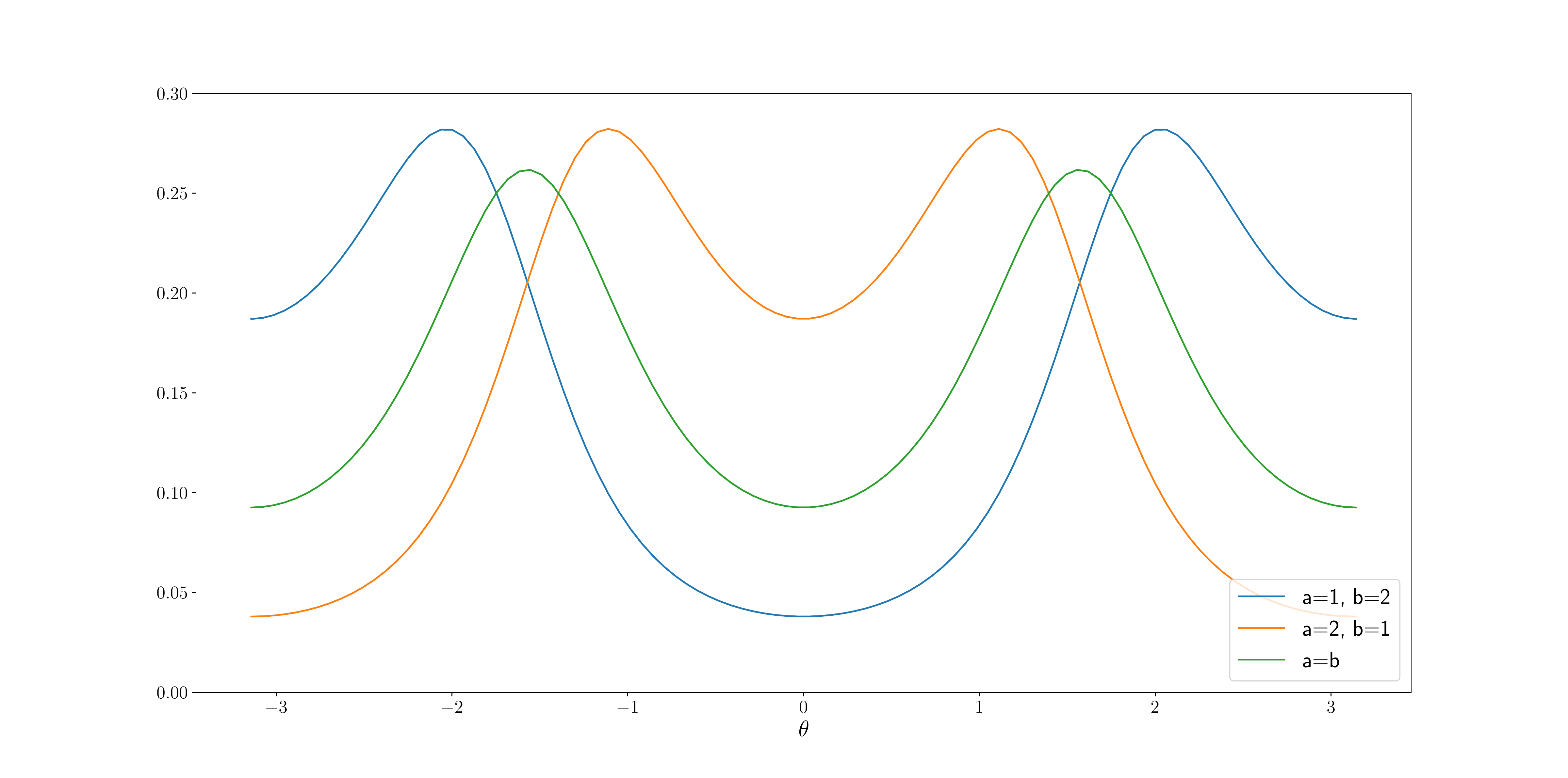}
\vspace{-2em}\caption{The density of the invariant probability distribution for Example \ref{subsec_thqubitdiff}.}
\label{figure_thqubitdiff}
\end{figure}

\subsection{Thermal qubit, jump case}
Our third example is the second one where the stimulating coherent field has relatively small amplitude and is filtered out. Then, the signal is composed only of the photons absorbed or emitted by the qubit. The resulting trajectory involves only jumps related to these events. The parameters defining the model are then, $H=0$, $I_b=\emptyset$ and $I_p=\{0,1\}$, $C_0=\sqrt{a}\,\sigma_+$ and $C_1=\sqrt{b}\,\sigma_-$ with $a,b\in\mathbb R_+\setminus\{0\}$.

The process $(\pi_{\hat x_t})_t$ is solution of 
\begin{equation}\label{eq:SDE12}
\begin{split}
\dd\pi_{\hat x_t}=&a\Big(\sigma_+\pi_{\hat x_{t-}}\sigma_--\tfrac{1}{4}\big((\Id-\sigma_z)\pi_{\hat x_{t-}}+\pi_{\hat x_{t-}}(\Id-\sigma_z)\big)\Big)\dd t\\
&+b\Big(\sigma_-\pi_{\hat x_{t-}}\sigma_+-\tfrac{1}{4}\big((\Id+\sigma_z)\pi_{\hat x_{t-}}+\pi_{\hat x_{t-}}(\Id+\sigma_z)\big)\Big)\dd t\\
& +\Big(\frac{\sigma_+\pi_{\hat x_{t-}}\sigma_-}{\tr(\sigma_-\sigma_+\pi_{\hat x_{t-}})}-\pi_{\hat x_{t-}}\Big)\Big(\dd N_0(t)-a\tr(\sigma_-\sigma_+\pi_{\hat x_{t-}})\dd t\Big)\\
&+\Big(\frac{\sigma_-\pi_{\hat x_{t-}}\sigma_+}{\tr(\sigma_+\sigma_-\pi_{\hat x_{t-}})}-\pi_{\hat x_{t-}}\Big)\Big(\dd N_1(t)-b\tr(\sigma_+\sigma_-\pi_{\hat x_{t-}})\dd t\Big)
\end{split}
\end{equation}
where $N_0$ and $N_1$ are Poisson processes of stochastic intensities 
\[t\mapsto\int_0^t a\tr(\sigma_-\sigma_+\pi_{\hat x_{s-}})\, \dd s\quad\textrm{and}\quad t\mapsto\int_0^t b\tr(\sigma_+\sigma_-\pi_{\hat x_{s-}})\,\dd s.\]
Assumptions \pur\ and {\lerg} hold as in Example \ref{subsec_thqubitdiff}.
\begin{prop}
Let $\{e_1,e_2\}$ denote the canonical basis of $\C^2$. The invariant measure for Equation \eqref{eq:SDE12} is
\[\mu\inv{}=\frac{a}{a+b}\,\delta_{\pi_{\hat e_1}}+\frac{b}{a+b}\,\delta_{\pi_{\hat e_2}}\]
\end{prop}
\begin{proof}
It is enough to check from \eqref{eq:SDE12} that, if ${\hat x_0}$ is either $\hat e_1$ or $\hat e_2$ then $(\pi_{\hat x_t})_t$ is a jump process on $(\pi_{\hat e_1},\pi_{\hat e_2})$ with intensity $b$ for the jumps from $\pi_{\hat e_1}$ to $\pi_{\hat e_2}$ and intensity $a$ for the reverse jumps.
\end{proof}

\subsection{Finite state space Markov process embedding}\label{sec:classical_MC}

In this last example we show how we can recover all the usual continuous-time Markov chains using special quantum trajectories.

Let $\{e_\ell\}_{\ell=1}^k$ be an orthonormal basis of $\C^k$, and $(X_t)_t$ a $\{e_1,\ldots,e_k\}$-valued Markov process with generator $Q$ (we recall that $Q$ is a $k\times k$ real matrix such that $\mathbb E(\langle v,X_t\rangle|X_0)=\langle v, e^{tQ}X_0\rangle$ for any vector $v\in \C^k$). Let $H$ be diagonal in the basis $\{e_\ell\}_{\ell=1}^k$, let $I_b=\emptyset$ and $I_p=\{ (i,j); i\neq j \mbox{ in }1,\ldots,k\}$ and for any $(i,j)\in I_p$ let $C_{i,j}=\sqrt{Q_{i,j}} \,e_j e_i^*$.

\begin{prop}\label{prop:compare_classical}
Let $(\hat x_t)_t$ be the quantum trajectory defined by Equation \eqref{eq:SDEpur} and the above parameters. Then assumption \pur\ holds. In addition,
\begin{enumerate}[label=(\roman*)]
\item Let $T=\inf\big\{t\geq 0: \hat x_t\in \{\hat e_1,\ldots,\hat e_k\}\big\}$. If for all $i$ there exists $j$ with $Q_{i,j}>0$ then for any probability measure $\mu$ over $\pc k$,  $\pp_\mu(T<\infty)=1$.\label{it:class-pur}
\item Conditionally on $\hat x_0\in \{\hat e_\ell\}_{\ell=1}^k$, the process $(\hat x_t)_t$ has the same distribution as the image by~$x\mapsto\hat x$ of $(X_t)_t$.
\item The assumption {\lerg} holds if and only if $(X_t)_t$ accepts a unique invariant measure. In that case, the unique invariant measure $\nu\inv$ for $(\hat x_t)_t$ is the image by $x\mapsto \hat x$ of the unique invariant measure for $(X_t)_t$.\label{it:class-erg}
\end{enumerate}
\end{prop}

\begin{proof}
Note first that any $C_{i,j}^*C_{i,j}=Q_{i,j}\,e_i e_i^*$
, so that \pur\ holds trivially. 

To prove (i), let $T_1=\inf\{t>0; \ \exists (i,j)\in I_p \mbox{ such that } N_{i,j}(t)>0\}$. Remark that because $\tr(C_{i,j}\pi_{\hat x_{s-}}C_{i,j}^*)=Q_{i,j} |\langle e_i, x_{s-}\rangle|^2$, the sum $\sum_{i,j}N_{i,j}$ of independent Poisson processes has intensity
\[\sum_{i,j}\int_0^t Q_{i,j}  |\langle e_i, x_{s-}\rangle|^2\, \dd s\geq  t \,\min_i Q_i\]
where $Q_i=\sum_j Q_{i,j}$ is positive by assumption, so that $T_1$ is almost surely finite. Now consider the almost surely unique $(i,j)$ in $I_p$ such that $N_{i,j}(T_1)>0$; necessarily $\tr(C_{i,j} \pi_{\hat x_{t-}} C_{i,j}^*)>0$, and then $\frac{C_{i,j} \pi_{\hat x_{t-}} C_{i,j}^*}{\tr(C_{i,j} \pi_{\hat x_{t-}} C_{i,j}^*)}=\pi_{\hat e_j}$, so that $T\leq T_1$. This proves (i).

Now, to prove (ii), remark that equation \eqref{eq1:SME} can be rewritten in the form
\begin{equation}
\begin{aligned}
\dd \pi_{\hat x_t}=&\hphantom{+}\sum_{(i,j)\in I_p}\big(\tr(C_{i,j}\pi_{\hat x_{t-}}C_{i,j}^*)\pi_{\hat x_{t-}}-\frac{1}{2}\{C_{i,j}^*C_{i,j},\pi_{\hat x_{t-}}\}\big)\,\dd t\nonumber\\
&+\sum_{(i,j)\in I_p}\Big(\frac{C_{i,j}\pi_{\hat x_{t-}}C_{i,j}^*}{\tr(C_{i,j}\pi_{\hat x_{t-}}C_{i,j}^*)}-\pi_{\hat x_{t-}}\Big)\,\dd N_{i,j}(t).
\end{aligned}
\end{equation}
Let $T_1$ be defined as above; then for $t<T_1$ the process $(\pi_{\hat x_t})_t$ satisfies
\begin{equation}\label{eq:ode}
\pi_{\hat x_t}=\pi_{\hat x_0}+\sum_{(i,j)\in I_p} \int_0^t \big(\tr(C_{i,j}\pi_{\hat x_{s-}}C_{i,j}^*)\pi_{\hat x_{s-}}-\frac{1}{2}\{C_{i,j}^*C_{i,j},\pi_{\hat x_{s-}}\}\big)\,\dd s.\end{equation}
Starting with an initial condition $\hat x_0\in \{\hat e_\ell\}_{\ell=1}^k$,
one proves easily that the integrand is zero, which means that $\pi_{\hat x_t}=\pi_{\hat x_0}$ for $t<T_1$. This shows in addition that for $t<T_1$, the intensity of $N_{i,j}$ is 
\[\int_0^t \tr(C_{i,j}\pi_{\hat x_{s-}}C_{i,j}^*)\,\dd s=
\left\{
\begin{array}{cl}
	Q_{i,j}\, t & \mbox{ if } x_0=e_i\\
	0		  & \mbox{ otherwise.}
\end{array}
\right.\]
Therefore, conditionally on $x_0=e_i$, $T_1=\inf\{t>0; \ \exists j\neq i \mbox{ such that } N_{i,j}(t)>0\}$ and there exists an almost surely unique $j$ such that $N_{i,j}(T_1)>0$. One then has
\[\pi_{\hat x_{T_1}}=\frac{C_{i,j} \pi_{\hat x_{T_1-}} C_{i,j}^*}{\tr(C_{i,j} \pi_{\hat x_{T_1-}} C_{i,j}^*)}=\pi_{\hat{e}_j}.\]
This shows that for $t\in[0,T_1]$ the process $(\hat x_t)_t$ has the same distribution as the process of equivalence classes of $X_t$. This extends to all $t$ by the Markov property of the Poisson processes. This proves (ii).

Points (i) and (ii) show that for $t>T_1$, the process $(\hat x_t)_t$ has the same distribution as $(X_t)_t$ with initial condition $X_{T_1}$ satisfying $\hat X_{T_1}=\hat x_{T_1}$. Therefore any invariant measure for $(\hat x_t)_t$ is the image by $x\mapsto \hat x$ of an invariant measure for $(X_t)_t$. Theorem \ref{thm:expo_conv_wasser1} and Section \ref{app:inv_measure} show that $(\hat x_t)_t$ admits at least one invariant measure, and that the invariant measure is unique if and only if \lerg\ holds. This implies that $(X_t)_t$ has a unique invariant measure if and only if \lerg\ holds.
\end{proof}

\paragraph{\textbf{Acknowledgments}} 
The research of T.B., Y.P.\ and C.P.\ has been supported by the ANR project StoQ ANR-14-CE25-0003-01. The research of T.B.\ has been supported by ANR-11-LABX-0040-CIMI within the program ANR-11-IDEX-0002-02. The research of M.F.\ was supported in part by funding from the Simons Foundation and the Centre de Recherches Mathématiques, through the Simons-CRM scholar-in-residence program. Y.P.\ acknowledges the support of ANR project NonStops ANR-17-CE40-0006, and of the Cantab Capital Institute for the Mathematics of Information at the University of Cambridge.

\bibliographystyle{abbrv}

\begin{thebibliography}{10}
	
	\bibitem{attal2006repeated}
	S.~Attal and Y.~Pautrat.
	\newblock From repeated to continuous quantum interactions.
	\newblock {\em Annales Henri Poincar{\'e}}, 7(1):59--104, 2006.
	
	\bibitem{BCFFS}
	M.~Ballesteros, N.~Crawford, M.~Fraas, J.~Fr{\"o}hlich, and B.~Schubnel.
	\newblock Perturbation theory for weak measurements in quantum mechanics,
	systems with finite-dimensional state space.
	\newblock {\em Annales Henri Poincar{\'e}}, 20(1):299--335, 2019.
	
	\bibitem{Ba5}
	A.~Barchielli.
	\newblock Continual measurements in quantum mechanics and quantum stochastic
	calculus.
	\newblock In {\em Open quantum systems. {III}}, volume 1882 of {\em Lecture
		Notes in Math.}, pages 207--292. Springer, Berlin, 2006.
	
	\bibitem{Ba2}
	A.~Barchielli and V.~P. Belavkin.
	\newblock Measurements continuous in time and a posteriori states in quantum
	mechanics.
	\newblock {\em J. Phys. A}, 24(7):1495--1514, 1991.
	
	\bibitem{BaGre}
	A.~Barchielli and M.~Gregoratti.
	\newblock {\em Quantum trajectories and measurements in continuous time},
	volume 782 of {\em Lecture Notes in Physics}.
	\newblock Springer, Heidelberg, 2009.
	\newblock The diffusive case.
	
	\bibitem{Ba3}
	A.~Barchielli and A.~S. Holevo.
	\newblock Constructing quantum measurement processes via classical stochastic
	calculus.
	\newblock {\em Stochastic Process. Appl.}, 58(2):293--317, 1995.
	
	\bibitem{Ba4}
	A.~Barchielli and A.~M. Paganoni.
	\newblock On the asymptotic behaviour of some stochastic differential equations
	for quantum states.
	\newblock {\em Infin. Dimens. Anal. Quantum Probab. Relat. Top.},
	6(2):223--243, 2003.
	
	\bibitem{BBT15}
	M.~Bauer, D.~Bernard, and A.~Tilloy.
	\newblock Computing the rates of measurement-induced quantum jumps.
	\newblock {\em Journal of Physics A: Mathematical and Theoretical},
	48(25):25FT02, 2015.
	
	\bibitem{Baumgartner}
	B.~Baumgartner and H.~Narnhofer.
	\newblock The structures of state space concerning quantum dynamical
	semigroups.
	\newblock {\em Rev. Math. Phys.}, 24(02):1250001, 2012.
	
	\bibitem{Belavkin1992}
	V.~P. Belavkin.
	\newblock Quantum stochastic calculus and quantum nonlinear filtering.
	\newblock {\em J. Multivariate Anal.}, 42(2):171--201, 1992.
	
	\bibitem{Inv1}
	T.~Benoist, M.~Fraas, Y.~Pautrat, and C.~Pellegrini.
	\newblock Invariant measure for quantum trajectories.
	\newblock {\em Probability Theory and Related Fields}, Jul 2018.
	
	\bibitem{BCCNP}
	C.~Bernardin, R.~Chetrite, R.~Chhaibi, J.~Najnudel, and C.~Pellegrini.
	\newblock Spiking and collapsing in large noise limits of sde's.
	\newblock {\em arXiv preprint arXiv:1810.05629}, 2018.
	
	\bibitem{bouten_quantum_2006}
	L.~Bouten and R.~van Handel.
	\newblock Quantum filtering: a reference probability approach.
	\newblock arXiv:math-ph/0508006, 2006.
	
	\bibitem{BreuerPetru}
	H.-P. Breuer and F.~Petruccione.
	\newblock {\em The theory of open quantum systems}.
	\newblock Oxford University Press, New York, 2002.
	
	\bibitem{Ca1}
	H.~Carmichael.
	\newblock {\em An Open Systems Approach to Quantum Optics: Lectures Presented
		at the Universit{\'e} Libre de Bruxelles, October 28 to November 4, 1991}.
	\newblock Springer Science, Jan. 1993.
	
	\bibitem{dalibard_wave-function_1992}
	J.~Dalibard, Y.~Castin, and K.~M{\o}lmer.
	\newblock Wave-function approach to dissipative processes in quantum optics.
	\newblock {\em Physical Review Letters}, 68(5):580--583, Feb. 1992.
	
	\bibitem{davies_markovian_1974}
	E.~B. Davies.
	\newblock Markovian master equations.
	\newblock {\em Communications in Mathematical Physics}, 39(2):91--110, June
	1974.
	
	\bibitem{davies_markovian_1976}
	E.~B. Davies.
	\newblock Markovian master equations. {II}.
	\newblock {\em Mathematische Annalen}, 219(2):147--158, June 1976.
	
	\bibitem{diosi_quantum_1988}
	L.~Diosi.
	\newblock Quantum stochastic processes as models for state vector reduction.
	\newblock {\em Journal of Physics A: Mathematical and General},
	21(13):2885--2898, July 1988.
	
	\bibitem{ficheux2018dynamics}
	Q.~Ficheux, S.~Jezouin, Z.~Leghtas, and B.~Huard.
	\newblock Dynamics of a qubit while simultaneously monitoring its relaxation
	and dephasing.
	\newblock {\em Nature communications}, 9(1):1926, 2018.
	
	\bibitem{Ga1}
	C.~W. Gardiner and P.~Zoller.
	\newblock {\em Quantum noise}.
	\newblock Springer Series in Synergetics. Springer-Verlag, Berlin, third
	edition, 2004.
	\newblock A handbook of Markovian and non-Markovian quantum stochastic methods
	with applications to quantum optics.
	
	\bibitem{gisin_quantum_1984}
	N.~Gisin.
	\newblock Quantum {Measurements} and {Stochastic} {Processes}.
	\newblock {\em Physical Review Letters}, 52(19):1657--1660, May 1984.
	
	\bibitem{GKS}
	V.~Gorini, A.~Kossakowski, and E.~C.~G. Sudarshan.
	\newblock Completely positive dynamical semigroups of {$N$}-level systems.
	\newblock {\em J. Mathematical Phys.}, 17(5):821--825, 1976.
	
	\bibitem{Ha1}
	S.~Haroche and J.-M. Raimond.
	\newblock {\em Exploring the quantum}.
	\newblock Oxford Graduate Texts. Oxford University Press, Oxford, 2006.
	\newblock Atoms, cavities and photons.
	
	\bibitem{Partha1}
	R.~L. Hudson and K.~R. Parthasarathy.
	\newblock Quantum {I}to's formula and stochastic evolutions.
	\newblock {\em Comm. Math. Phys.}, 93(3):301--323, 1984.
	
	\bibitem{JacodShiryaev}
	J.~Jacod and A.~N. Shiryaev.
	\newblock {\em Limit theorems for stochastic processes}, volume 288 of {\em
		Grundlehren der Mathematischen Wissenschaften [Fundamental Principles of
		Mathematical Sciences]}.
	\newblock Springer-Verlag, Berlin, second edition, 2003.
	
	\bibitem{Kabanov}
	J.~M. Kabanov, R.~{\v{S}}. Lipcer, and A.~{\v{S}}irjaev.
	\newblock Absolute continuity and singularity of locally absolutely continuous
	probability distributions. i.
	\newblock {\em Mathematics of the USSR-Sbornik}, 35(5):631, 1979.
	
	\bibitem{MR611513}
	S.~Karlin and H.~M. Taylor.
	\newblock {\em A second course in stochastic processes}.
	\newblock Academic Press, Inc. [Harcourt Brace Jovanovich, Publishers], New
	York-London, 1981.
	
	\bibitem{KuMa}
	B.~K{\"u}mmerer and H.~Maassen.
	\newblock A pathwise ergodic theorem for quantum trajectories.
	\newblock {\em J. Phys. A}, 37(49):11889--11896, 2004.
	
	\bibitem{Li}
	G.~Lindblad.
	\newblock On the generators of quantum dynamical semigroups.
	\newblock {\em Comm. Math. Phys.}, 48(2):119--130, 1976.
	
	\bibitem{Maassen}
	H.~Maassen and B.~K{\"u}mmerer.
	\newblock Purification of quantum trajectories.
	\newblock {\em Lecture Notes-Monograph Series}, 48:252--261, 2006.
	
	\bibitem{BMLalgebra}
	S.~Mac~Lane and G.~Birkhoff.
	\newblock {\em Algebra}.
	\newblock Chelsea Publishing Co., New York, third edition, 1988.
	
	\bibitem{Pe1}
	C.~Pellegrini.
	\newblock Existence, uniqueness and approximation of a stochastic
	{S}chr\"{o}dinger equation: the diffusive case.
	\newblock {\em Ann. Probab.}, 36(6):2332--2353, 2008.
	
	\bibitem{Pe2}
	C.~Pellegrini.
	\newblock Poisson and diffusion approximation of stochastic master equations
	with control.
	\newblock {\em Ann. Henri Poincar\'{e}}, 10(5):995--1025, 2009.
	
	\bibitem{Pe3}
	C.~Pellegrini.
	\newblock Markov chains approximation of jump-diffusion stochastic master
	equations.
	\newblock {\em Ann. Inst. Henri Poincar\'{e} Probab. Stat.}, 46(4):924--948,
	2010.
	
	\bibitem{ticozziviola}
	F.~Ticozzi and L.~Viola.
	\newblock Quantum markovian subsystems: invariance, attractivity, and control.
	\newblock {\em IEEE Transactions on Automatic Control}, 53(9):2048--2063, 2008.
	
	\bibitem{wisemanmilburn}
	H.~M. Wiseman and G.~J. Milburn.
	\newblock {\em Quantum Measurement and Control}.
	\newblock Cambridge University Press, Jan. 2010.
	
	\bibitem{wolftour}
	M.~M. Wolf.
	\newblock Quantum channels \& operations: Guided tour.
	\newblock
	\url{http://www-m5.ma.tum.de/foswiki/pub/M5/Allgemeines/MichaelWolf/QChannelLecture.pdf},
	2012.
	\newblock Lecture notes based on a course given at the Niels-Bohr Institute.
	
\end{thebibliography}

\end{document}